%% file: main.tex
\documentclass[reprint,amsmath,amssymb,aps,superscriptaddress,nofootinbib,raggedbottom]{revtex4-2}
\usepackage{graphicx}
\usepackage{amsthm,mathtools,amsfonts}
\usepackage{newtxtext,newtxmath}
\usepackage{xcolor}
\usepackage{tikz-cd}
\tikzcdset{arrow style=math font}
\usetikzlibrary{fit,backgrounds,arrows.meta}
\usepackage{booktabs,ragged2e}
\usepackage[hypertexnames=false,colorlinks=true,citecolor=teal,linkcolor=teal,urlcolor=teal]{hyperref}

\usepackage{algpseudocode}

\newcounter{frcalgorithm}
\renewcommand{\thefrcalgorithm}{\arabic{frcalgorithm}}

\definecolor{productgreen}{HTML}{E3EBF1}
\input{newcommands}

\theoremstyle{plain}
\newtheorem{definition}{Definition}
\newtheorem{theorem}[definition]{Theorem}
\newtheorem{lemma}[definition]{Lemma}
\newtheorem{proposition}[definition]{Proposition}
\newtheorem{corollary}[definition]{Corollary}

\makeatletter
\long\def\@makecaption#1#2{\vskip\abovecaptionskip
{\justifying\noindent #1.\ #2\par}\vskip\belowcaptionskip}
\makeatother

\begin{document}

\title{Ultra-high-distance quantum memories from amplified qLDPC codes}

\author{\mbox{Zijian Liang}}
\thanks{These authors contributed equally to this work. \\ Ordering determined by a random number generator.}
\affiliation{International Center for Quantum Materials, School of Physics, Peking University, 100871 Beijing, China}

\author{\mbox{Boren Gu}}
\thanks{These authors contributed equally to this work. \\ Ordering determined by a random number generator.}
\affiliation{Dahlem Center for Complex Quantum Systems, Freie Universität Berlin, 14195 Berlin, Germany}

\author{\mbox{Yu-An Chen}}
\email{yuanchen@pku.edu.cn}
\affiliation{International Center for Quantum Materials, School of Physics, Peking University, 100871 Beijing, China}

\author{\mbox{Jens Eisert}}
\email{jense@zedat.fu-berlin.de}
\affiliation{Dahlem Center for Complex Quantum Systems, Freie Universität Berlin, 14195 Berlin, Germany}

\author{\mbox{Zongyuan Wang}}
\email{zongyuan.wang15@outlook.com}
\affiliation{International Center for Quantum Materials, School of Physics, Peking University, 100871 Beijing, China}

\date{\today}

\begin{abstract}
A large code distance does not by itself guarantee a well-protected quantum memory, as faults during syndrome extraction can propagate into correlated data errors. Certifying circuit distance becomes demanding as circuits grow. In this work, we introduce \emph{distance amplifiers}, a modular approach to increasing both code distance and certified circuit-level protection in quantum low-density parity-check codes. The construction tensors a base code with an amplifier, adapting their syndrome-extraction circuits to the amplified code. To certify the resulting memory, we develop a fault-response certificate that yields rigorous lower bounds on circuit distance. Suitable gate ordering and flag qubits control correlated errors, allowing us to establish explicit conditions under which the amplified circuit distance equals the product of the constituent circuit distances. We prove this multiplicative relation for three types of amplifiers: the flagged $\code{4,2,2}$, hypergraph-product, and rotated surface codes. The certificate extends recursively, so the circuit distance multiplies exactly at every amplification step. For example, four successive applications of the flagged $\code{4,2,2}$ amplifier to a $\code{18,4,4}$ seed yield a $\code{13320,64,64}$ code with the circuit distance $d_{\mathrm{circ}} \geq 48$ and check weights $w \leq 14$. We further show how individual logical qubits can be tracked explicitly through amplification. By enabling high-distance quantum memories to be built and certified from smaller building blocks, our framework offers a systematic route toward scalable fault-tolerant quantum computing.
\end{abstract}

\maketitle

\section*{Introduction}
\label{sec:introduction}

A central challenge in building scalable quantum computers is to develop error-correction and fault-tolerance methods capable of sustaining long computations in the presence of noise~\cite{Kitaev-AnnPhys-2003,GottesmanNewReview,QECBasic,Roads,RevModPhys.87.307,GottesmanThesis,NielsenChuang,MindTheGaps}. Substantial experimental progress and theoretical breakthroughs have made \emph{quantum low-density parity-check} (qLDPC) codes~\cite{PRXQuantum.2.040101,9996782, Panteleev2021degeneratequantum} an increasingly promising route towards this goal. Asymptotically good qLDPC constructions achieve constant encoding rate and linear distance~\cite{9996782,VidickLDPC,9567703}, while other code families can offer stronger performance at the finite block sizes relevant to near-term implementations~\cite{zhao2026ultrahighratequantumerrorcorrection,bravyi2024high,QGPU,BhardwajMaMeister2026}. Realizing these potential resource savings in practice requires codes to be co-designed with efficient syndrome-extraction circuits and logical-operation protocols~\cite{gidney2021factor2,Pinnacle,Preskill10000,OurPlanarArchitecture}.

Code distance alone, however, does not determine the protection of a quantum memory. During syndrome extraction, a single circuit fault can propagate into errors on multiple data qubits. The \emph{circuit distance}, the minimum number of faults producing an undetected logical error, can therefore be smaller than the code distance~\cite{DetectorErrorModels,ManesClaes2025}. Certifying circuit distance by direct numerical search also becomes increasingly demanding as the extraction circuit grows~\cite{webster2026distancefindingalgorithmsquantumcodes}.

\begin{figure*}[ht]
\centering
\includegraphics[width=\textwidth]{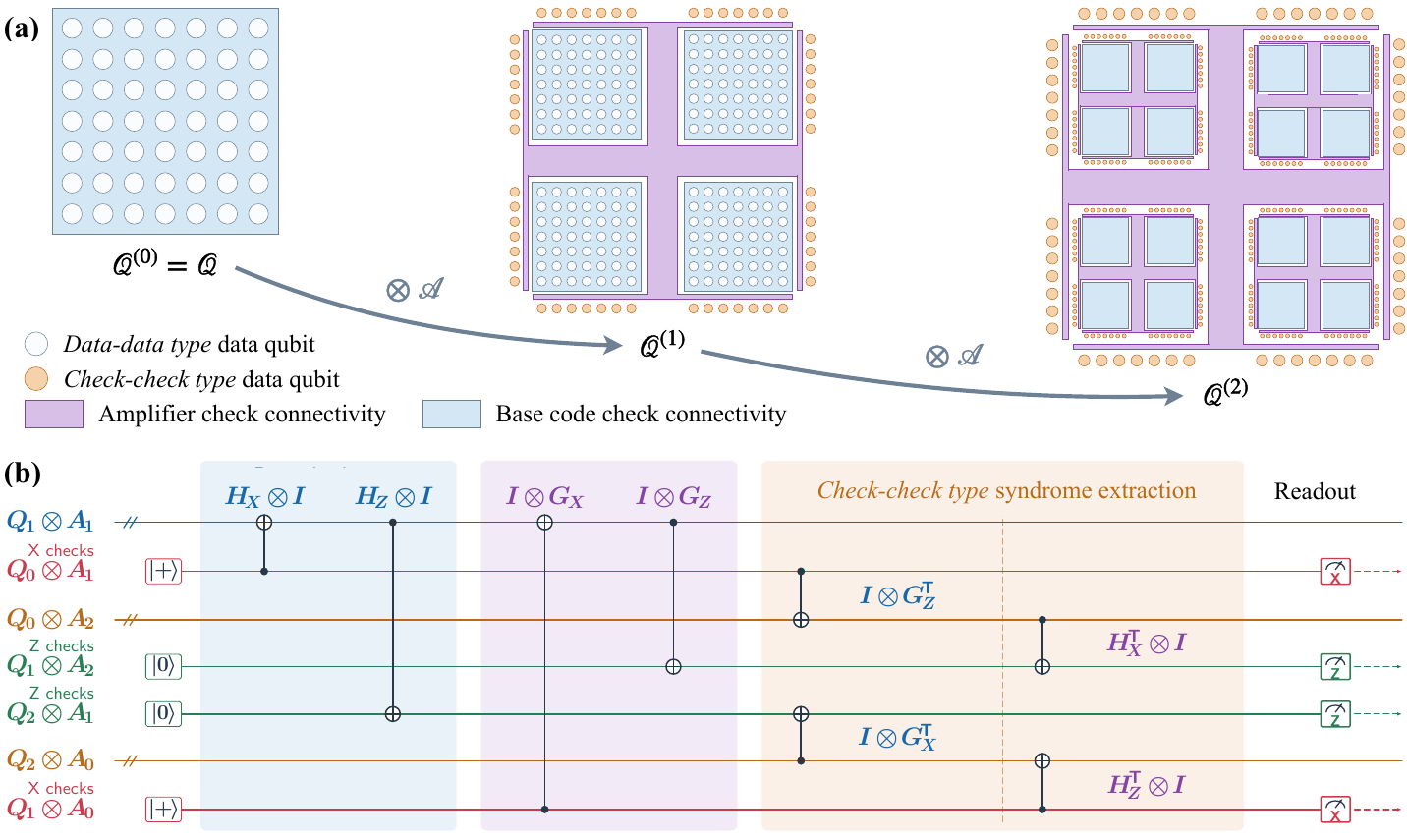}
\caption{
\textbf{Recursive amplification and staged joint syndrome extraction.}
\textbf{(a) Successive amplification.} Starting from $\QQQ^{(0)}=\QQQ$, each step applies the same amplifier according to $\QQQ^{(t+1)}=\QQQ^{(t)}\otimes\AAA$. Blue and purple shading denote check connectivity inherited from the base code and amplifier, respectively. White circles are data qubits in the central (data--data) register $Q_1\otimes A_1$. Orange circles are data qubits in the two side (check--check) registers, which complete the inherited and coupling checks while preserving commutation; the three data registers are defined in Eq.~\eqref{eq:tensor-product-complex}.
\textbf{(b) One round of staged joint syndrome extraction.} Copies of the base-code extraction circuit run in parallel on central slices, followed by parallel copies of the amplifier extraction circuit. Two ordered side-coupling windows then complete the inherited base checks and the new coupling checks, respectively. Flags, when present, are retained within the corresponding constituent circuits and measured before the appended side couplings. Syndrome auxiliary systems are read out after both side-coupling windows.}
\label{fig:css_tensor_overview}
\end{figure*}

In this work, we introduce \emph{distance amplifiers}, a modular approach to increasing both code distance and certified circuit-level protection. The construction uses two \emph{Calderbank-Shor-Steane} (CSS) codes~\cite{CalderbankShor,Steane1996}: a base code $\QQQ$ and a typically small amplifier code $\AAA$. We take one copy of $\QQQ$ for each qubit of $\AAA$. The checks of $\AAA$ then determine how corresponding data qubits in different copies are coupled~\cite{zhang2026coupledlayerconstructionquantumproduct,AudouxCouvreur2019}. We extend both the inherited base checks and the new coupling checks to act on additional data qubits so that all checks commute (Fig.~\ref{fig:css_tensor_overview}a). The resulting code is a CSS tensor-product code, whose distance gains can be certified using the logical-overlap criterion~\cite{AudouxCouvreur2019}. For example, the $\code{4,2,2}$ amplifier doubles both the number of logical qubits and the code distance, mapping a base code with parameters $\code{n,k,d}$ to
\begin{equation}
    \code{4n+m_X+m_Z,2k,2d},
    \label{eq:422_amplification}
\end{equation}
where $m_X$ and $m_Z$ are the numbers of $X$- and $Z$-checks, respectively, in the chosen base-code presentation.

Standard code concatenation expresses outer checks as products of inner-code logical operators~\cite{wills2026concatenatingalgebraiccodeshighrate,GidneyNewmanBrooks2025Yoked,low2026denserplanarsurfacecode}. Repeated logical-operator substitution can make check weights grow exponentially with concatenation depth. Our construction instead couples codes through physical checks. The tensor-product structure gives additive bounds on check weights and qubit degrees~\cite{AudouxCouvreur2019}, so both grow at most linearly with amplification depth for a fixed amplifier.

For syndrome extraction, we extend copies of the base-code and amplifier circuits with physical CNOT gates to the additional data qubits, completing the product checks without requiring logical gates on the base code. These new couplings introduce additional paths for fault propagation. We control the resulting correlated errors through suitable gate ordering or flag qubits within a \emph{staged joint} (SJ) extraction schedule (Fig.~\ref{fig:css_tensor_overview}b).

Our analysis begins by identifying the data-error patterns that faults can produce in individual check-measurement circuits. We introduce a \emph{fault-response certificate}, which lower-bounds the total assigned fault count needed to combine these local responses into a nontrivial logical error. In this counting problem, ordinary noisy syndrome outcomes are left unrestricted, while the flag constraints are retained. Every undetected logical fault history is covered by this analysis, so the certificate also gives a lower bound on circuit distance.

We choose amplifiers for which every nontrivial logical error of the amplified code yields several nontrivial base-code logical responses. Each response is obtained by summing the error patterns modulo two across a selected group of copies, with no copy used in more than one group. These responses need not arise from complete undetected histories of the base circuit, so we bound their fault costs using the base certificate rather than circuit distance alone. The amplifier's response properties and the prescribed extraction schedule allow these costs to be added without double-counting. We express the resulting multiplication factor as a \emph{transferable amplifier certificate}: multiplying it by the base certificate gives a certificate for the amplified circuit. This new certificate has the same form as the base certificate and can therefore be reused recursively, without a new search over the fault configurations of each amplified circuit.

We verify transferable certificates equal to the amplifier code distance in each Pauli sector for three families. A flagged $\code{4,2,2}$ amplifier~\cite{Knill2005,Koh2026PhantomCodes} detects dangerous correlated errors and doubles the base certificate. Canonical \textit{hypergraph-product} (HGP) amplifiers~\cite{TillichZemor2014,ManesClaes2025,quintavalle_reshape_2022} use their row and column structure to separate fault contributions, while rotated surface amplifiers~\cite{GoogleThreshold,PablaWangHong2026} control correlated errors through prescribed gate ordering without introducing new flags. For these families, when the seed certificate equals the seed circuit distance in a given Pauli sector, a matching upper bound establishes exact circuit-distance multiplication at every step under the prescribed extraction schedule.

To illustrate the finite-size gains, we start from a $\code{18,4,4}$ bivariate-bicycle code~\cite{bravyi2024high,liang2025generalized,Wang_2026} whose extraction circuit has distance $4$ and a fault-response certificate of $3$. A single amplification step with a $\code{361,1,19}$ rotated surface code yields a $\code{9738,4,76}$ memory with a certified bound $d_{\mathrm{circ}}\geq57$ and check weights at most $8$, retaining all four logical qubits. Alternatively, four successive applications of the flagged $\code{4,2,2}$ amplifier yield a $\code{13320,64,64}$ memory with a certified bound $d_{\mathrm{circ}}\geq48$ and check weights at most $14$. We compare different amplifiers in terms of their certified distance gains and data-qubit overheads.

Distance amplification also allows individual logical qubits and their Pauli operators to be tracked explicitly through recursion. Recent progress in qLDPC surgery highlights the value of explicit logical addresses for fast and parallel measurements~\cite{LiftedSurgery, XuZhouZheng2024Homological, QGPU, ZhengZhengJiangXu2026CanonicalLP, ZhengJiangXu2025HighRateSurgery, CowtanHeWilliamson2026}. Starting from dual logical bases of the constituent codes, we construct a tensor-product logical basis with inherited labels and preserved $X$--$Z$ pairing. Each chosen logical representative acts on copies of its base-code support selected by the corresponding amplifier representative, making its physical support explicit. We also show how compatible SWAP-based automorphisms~\cite{sayginel2025automorphisms} and algebraic surgery maps~\cite{IdeGowda2025} lift to the amplified code. Distance amplifiers thus combine recursive circuit-level certification with an explicit logical structure inherited from smaller building blocks.

\section*{Results}

\subsection{Distance amplification for qLDPC memories}

We first establish the amplified code parameters and distance gains. Let the CSS base code $\QQQ$ encode $k$ logical qubits in $n$ data qubits, and let the CSS amplifier $\AAA$ encode $k_A$ logical qubits in $n_A$ data qubits. The numbers of retained $X$- and $Z$-check rows are $m_X,m_Z$ for the base code and $m_X^A,m_Z^A$ for the amplifier. Throughout, we assume that each amplifier check matrix has full row rank. The base-code checks need not be independent.

We form the amplified code $\TTT=\QQQ\otimes\AAA$ by taking the central three-term truncation of the tensor complex~\cite{AudouxCouvreur2019}, as specified in Eq.~\eqref{eq:tensor-product-complex}. Its data occupy one \textit{central (data--data) register}, indexed by pairs of constituent data qubits, and two \textit{side (check--check) registers}, indexed by pairs of base-code and amplifier checks of opposite Pauli types (Fig.~\ref{fig:css_tensor_overview}a). Writing $\widetilde n$ and $\widetilde k$ for the numbers of data and logical qubits in the amplified code, we obtain
\begin{equation}
    \widetilde n=nn_A+m_Zm_X^A+m_Xm_Z^A,
    \qquad
    \widetilde k=kk_A.
\end{equation}
The three terms in $\widetilde n$ count the central register and the two side registers, respectively.

For $P\in\{X,Z\}$, let $d_P$ denote the minimum weight of a nontrivial $P$-type logical operator of the base code, where weight counts the data qubits on which the operator acts nontrivially. Define $d_P^A$ and $\widetilde d_P$ analogously for the amplifier and amplified code. The usual base-code distance is $d=\min\{d_X,d_Z\}$.

The logical-overlap criterion~\cite{AudouxCouvreur2019}, reviewed in Appendix~\ref{app:distance-amplification}, provides certified lower bounds on the distance gain. We call $\AAA$ an $(\alpha_X,\alpha_Z)$-distance amplifier, with $\alpha_X,\alpha_Z\geq1$, if it guarantees
\begin{equation}
    \left\lceil\alpha_P d_P\right\rceil
    \leq \widetilde d_P
    \leq d_P d_P^A,
    \qquad P\in\{X,Z\},
    \label{eq:main-static-bounds}
\end{equation}
for every CSS base code with $k>0$. Here, $\alpha_P$ is the certified multiplicative gain in the $P$ sector, and $\lceil\cdot\rceil$ denotes rounding up to an integer. The upper bound follows from tensoring minimum-weight logical operators of the two constituent codes. The certified gain $\alpha_P$ need not equal the amplifier distance $d_P^A$; when they coincide, the bounds establish exact code-distance multiplication.

Using the same fixed amplifier at each step, we define $\QQQ^{(0)}=\QQQ$ and $\QQQ^{(t+1)}=\QQQ^{(t)}\otimes\AAA$, taking the central truncation and retaining the full resulting check presentation after each step. Here, $t$ counts amplification steps, and $n^{(t)}$, $k^{(t)}$, and $d_P^{(t)}$ denote the data-qubit count, logical-qubit count, and $P$-sector code distance of $\QQQ^{(t)}$, respectively. The logical-qubit count and distances satisfy
\begin{equation}
    k^{(t)}=k k_A^t,
    \qquad
    \left\lceil\alpha_P^t d_P\right\rceil
    \leq d_P^{(t)}
    \leq (d_P^A)^t d_P.
\end{equation}

To quantify sparsity, let $w(\mathcal{C})$ be the largest row or column weight of either check matrix of a CSS code $\mathcal{C}$ in the chosen presentation. A row weight counts the data qubits in a check, while a column weight counts the checks of a given Pauli type acting on a data qubit. The tensor-product structure gives $w(\QQQ^{(t)})\leq w(\QQQ)+t w(\AAA)$ (Methods~\ref{subsec:code-distance-amplifier}). Thus, check weights and qubit degrees grow at most linearly with amplification depth for a fixed amplifier. This upper bound need not be tight. In particular, applying a fixed number of amplification steps to a qLDPC family preserves the LDPC property.

Amplifier choice also determines the data-qubit cost of reaching a target distance. For equal certified gains $\alpha_X=\alpha_Z=\alpha_A>1$, define the one-step data-qubit multiplication factor $\eta_A(\QQQ)=\widetilde n/n$ and the score
\begin{equation}
    \beta_A^{(1)}(\QQQ)
    =\frac{\log\alpha_A}{\log\eta_A(\QQQ)}.
\end{equation}
This score compares the guaranteed distance gain with the data-qubit growth; its value depends on the retained base-code check counts as well as on the amplifier. Under repeated amplification with $m_X^A,m_Z^A>0$, the data-qubit count $n^{(t)}$ grows asymptotically with a factor $c_A=n_A+\sqrt{2m_X^A m_Z^A}$ per step. The corresponding asymptotic score is $\beta_A^{(\infty)}=\log\alpha_A/\log c_A$, which gives the power-law exponent of the certified code-distance lower bound as a function of the data-qubit count. The derivation and exact qubit-count recurrence are given in Methods~\ref{subsec:amplifier-benchmark}. These benchmarks count data qubits only; they do not include syndrome or flag auxiliary systems, gate counts, circuit depth, or decoder performance.

Table~\ref{tab:balanced-amplifiers} compares representative amplifiers using the reference base-code check counts $m_X=m_Z=n/2$. Among the listed candidates, the $\code{4,2,2}$ amplifier has the lowest one-step data-qubit overhead for certified distance doubling, while rotated surface amplifiers have the lowest listed overheads at certified gains of three, four, and five. For a specified finite distance target, the comparison should instead use the exact qubit counts and the integer-rounded distance guarantees.

These code-distance bounds do not by themselves guarantee the same protection during noisy syndrome extraction. We next analyze the \emph{staged joint} (SJ) extraction schedule and establish recursive circuit-distance guarantees for the three amplifier families introduced above.

\begin{figure}[ht]
\centering
\includegraphics[width=\columnwidth]{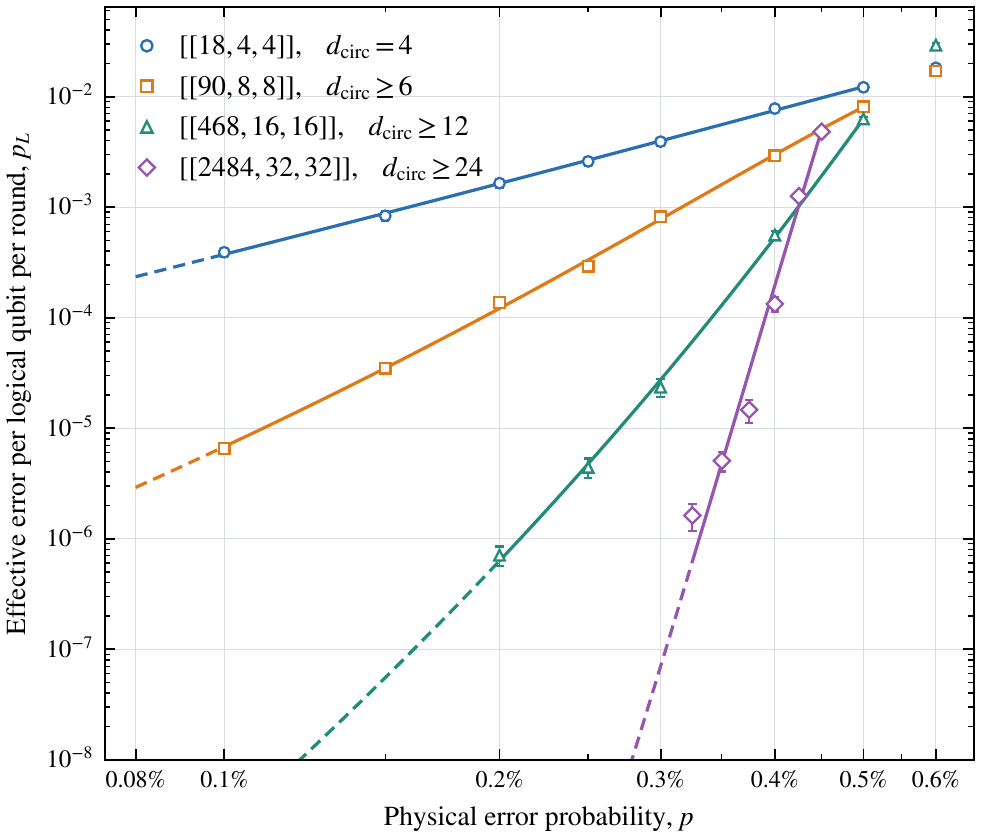}
\caption{\textbf{Circuit-level memory performance under recursive amplification.} Effective logical error rate $p_L$ per logical qubit per noisy round versus physical error probability $p$ for the $\code{18,4,4}$ seed and its first three flagged $\code{4,2,2}$ amplifications. Simulations use noiseless idling and independent BP+LSD decoding of the CSS sectors. Curves show empirical binomial maximum-likelihood fits and extrapolations using $p_L=p^{\delta/2}\exp(c_0+c_1p+c_2p^2)$~\cite{bravyi2024high}, with $\delta=4$ for the seed and $\delta=6,12,24$ for the three amplification levels, respectively, corresponding to the circuit-distance lower bounds shown in the legend.}
\label{fig:recursive-memory-performance}
\end{figure}

\subsection{Making the distance gain survive extraction}

We implement amplification using a \emph{staged joint} (SJ) syndrome-extraction schedule (Fig.~\ref{fig:css_tensor_overview}b; Methods~\ref{subsec:syndrome-circ-amplified}). Each round executes copies of the base-code circuits, followed by copies of the amplifier circuits and then two ordered windows of couplings to the new side data. The constituent gadgets retain their local gate orders and flag measurements. Flags are read before the newly appended couplings, while syndrome-auxiliary readouts are deferred until these couplings are complete.

For an extraction circuit $\mathcal S$, let $d_{\mathrm{circ},P}(\mathcal S)$ denote the minimum number of faulty physical locations in a pure-$P$ history producing a nontrivial final $P$ logical with zero complete measurement record, where $P\in\{X,Z\}$, i.e., the \textit{$P$-sector circuit distance}. The record includes every retained noisy check outcome, each individual flag outcome and the outcomes of the perfect final check measurements, relative to the fault-free experiment.

Our lower bounds use a \emph{fault-response certificate} $D_P$ for the base circuit. A measurement \textit{gadget} comprises the circuit for one stabilizer check, including its syndrome ancilla and any flags. The certificate asserts that every decomposition of a nontrivial $P$ logical into zero-flag fault-bundle responses from $P$-check gadgets and weight-one data responses has an assigned fault count of at least $D_P$ (Def.~\ref{def:joint-response-certificate} in Methods~\ref{subsec:circuit-distance-methods}). Each bundle has zero net response on every flag and is assigned its number of faulty locations; each weight-one data response is assigned a fault count of one. Ordinary syndrome outcomes are unrestricted, allowing the certificate to constrain logical responses whose allocated faults need not form complete undetected histories. Every undetected history admits such a decomposition, so $D_P\leq d_{\mathrm{circ},P}(\mathcal S)$, with equality not assumed.

For each of the three amplifier families below, a nontrivial product logical yields $d_P^A$ nontrivial base-code logical responses by \emph{contraction}: summing, modulo two, its central error patterns over selected copies of the base code. The prescribed amplifier circuits and SJ schedule allow fault responses to be allocated among these contractions without increasing their total assigned fault count. Each contraction requires at least $D_P$, giving a product lower bound. Conversely, copying a complete undetected base-code history along a minimum-weight amplifier logical produces an undetected product history with $d_P^A$ times as many faults. Under the circuit conditions specified in Methods~\ref{subsec:circuit-distance-methods} and Appendix~\ref{app:circuit_distance_amplification}, the amplified circuit $\widetilde{\mathcal S}$ therefore satisfies
\begin{equation}
    d_P^A D_P
    \leq d_{\mathrm{circ},P}(\widetilde{\mathcal S})
    \leq d_P^A d_{\mathrm{circ},P}(\mathcal S).
\label{eq:sj-circuit-distance-bounds}
\end{equation}
When the seed certificate is attained, $D_P=d_{\mathrm{circ},P}(\mathcal S)$, the bounds coincide with
\begin{equation}
    d_{\mathrm{circ},P}(\widetilde{\mathcal S}) =d_P^A d_{\mathrm{circ},P}(\mathcal S).
\label{eq:sj-circuit-distance-equality}
\end{equation}
The allocation argument also transfers the certificate itself, from $D_P$ to $d_P^A D_P$. Retaining the complete amplified circuit as the next base therefore gives recursive lower bounds and, for an attained seed certificate, exact multiplication at every compatible amplification step.

Certificates below the seed circuit distance still give substantial guarantees. For the $\code{18,4,4}$ memory with circuit distance $4$ and certificates $D_X=D_Z=3$, a distance-$19$ rotated surface amplifier gives $57\leq d_{\mathrm{circ}}\leq76$, while four flagged $\code{4,2,2}$ steps give $48\leq d_{\mathrm{circ}}\leq64$.

\subsubsection{The flagged four-copy distance amplifier}

The $\code{4,2,2}$ amplifier gives the simplest realization of fault-response certificate doubling and, for an attained seed certificate, circuit-distance doubling. Every nontrivial product logical admits two complementary pair contractions representing the same nontrivial base-code logical class. In a new amplifier gadget, a zero-flag fault bundle produces a response consisting of a component common to all four copies and a residual whose weight is no greater than the fault count. The common component cancels in both pair contractions, while the residual responses can be allocated between them without increasing their total assigned fault count (App.~\ref{app:Proof_for_flagged_422}). Thus a base certificate $D_P$ transfers to $2D_P$.

For a seed circuit $\mathcal S^{(0)}$ whose two sectors admit certificates at the common value $D_X^{(0)}=D_Z^{(0)}=d_{\mathrm{circ}}(\mathcal S^{(0)})$, recursive amplification yields
\begin{equation}
    d_{\mathrm{circ}}(\mathcal S^{(t)})=2^td_{\mathrm{circ}}(\mathcal S^{(0)}).
    \label{eq:flagged-circuit-distance-recursion}
\end{equation}
Circuit-level simulations of the $\code{18,4,4}$ base code memory and its first three flagged $\code{4,2,2}$ amplifications show increasingly steep logical-error suppression in the sampled physical error rate regime, from $0.1\%$ to $0.6\%$ (Fig.~\ref{fig:recursive-memory-performance}).

\subsubsection{Canonical HGP amplifiers}

Here, the row and column structure of an HGP code separates the fault responses that must be counted for distance amplification. Let $\AAA=\mathrm{HGP}(M_1,M_2)$, where the classical seed matrices have full row rank and nontrivial kernels of distances $d_1$ and $d_2$. The amplifier has $d_X^A=d_2$ and $d_Z^A=d_1$, and tensoring with a CSS base code gives the exact code distances $\widetilde d_X=d_2d_X$ and $\widetilde d_Z=d_1d_Z$~\cite{ZengPryadko2020}.

For the canonical check presentation and extraction circuits specified in Appendix~\ref{app:hgp-amplification}, every nontrivial product logical has $d_P^A$ nontrivial base-code contractions on disjoint amplifier slices. Each inherited zero-flag gadget response is assigned wholly to one slice, while each new amplifier fault contributes at most one weight-one data response to one selected slice. This gives Eq.~\eqref{eq:sj-circuit-distance-bounds} with gains $d_2$ and $d_1$ in the two sectors. No new flags are required, and any permitted base-circuit flags are retained. The slice structure underlying this allocation is closely related to the distance-preserving properties of canonical HGP syndrome extraction~\cite{ManesClaes2025,TanStambler2025}.

\subsubsection{Rotated surface amplifiers}

For surface-code amplifiers, disjoint logical strings provide the contractions, and the gate order controls their fault responses. A rotated rectangular surface amplifier $\AAA_{h,w}$ encodes one logical qubit and has $(d_X^A,d_Z^A)=(h,w)$. Its geometry makes the contractions explicit: a nontrivial product $X$ logical yields the same nontrivial base-code logical class in each of $h$ row contractions, while a product $Z$ logical gives the analogous contractions along $w$ columns.

The prescribed CNOT ordering orients two-qubit hooks transverse to the corresponding amplifier logical strings~\cite{GoogleThreshold}. After contraction, each new amplifier fault contributes at most one weight-one base-code response to a single row or column contraction. Each inherited zero-flag gadget bundle remains intact within a single row or column, so its assigned fault count is included at most once (App.~\ref{app:rotated-surface-code-amplifiers}). The SJ construction therefore transfers $(D_X,D_Z)$ to $(hD_X,wD_Z)$ without introducing new flags. When the corresponding seed certificates are attained, the circuit distances multiply exactly by $h$ and $w$. For a fixed square amplifier with parameters $\code{d_A^2,1,d_A}$, if both seed sectors admit a certificate at the common value $d_{\mathrm{circ}}(\mathcal S^{(0)})$, then the recursive circuit distance is $d_A^td_{\mathrm{circ}}(\mathcal S^{(0)})$.

\subsection{Logical addressability}
\label{sec:addressability}

The tensor construction also identifies where the amplified logical qubits reside. To track individual logical qubits through amplification, we choose dual $X$- and $Z$-logical bases of the base code $\QQQ$ and amplifier $\AAA$. Writing their representatives as $\ell_{P,\mu}$ and $\ell^A_{P,a}$, with $\mu=1,\ldots,k$ and $a=1,\ldots,k_A$, we construct
\begin{equation}
    \widetilde\ell_{P,\mu a}
    =\bigl(\ell_{P,\mu}\otimes\ell^A_{P,a},0,0\bigr),
    \qquad P\in\{X,Z\}.
\label{eq:canonical-basis-product}
\end{equation}
The first component acts on the central register, while the two zero components indicate identity on the side data. When both amplifier check matrices have full row rank, these representatives form a complete basis of $kk_A$ logical pairs, with $\widetilde\ell_{X,\mu a}^{\intercal}\widetilde\ell_{Z,\nu b}=\delta_{\mu\nu}\delta_{ab}$. The label $(\mu,a)$ therefore identifies an encoded qubit and its conjugate logical operators.

Each representative is supported on copies of the corresponding base-code region selected by the amplifier support. This carries supplied row--column~\cite{quintavalle_reshape_2022,quintavalle_partitioning_2023} or fiber~\cite{ZhengZhengJiangXu2026CanonicalLP} addresses through recursion, adding an amplifier label at each step when $k_A>1$. Base SWAP automorphisms also lift whenever their data permutations extend consistently to permutations of the retained check rows. The lifted physical permutation acts on all three data registers and induces the same base logical action at every amplifier label.

The inherited logical coordinates also support two algebraic lifts of base-code surgery. A complete lift imposes the base logical constraints at every amplifier label and retains the tensor structure of the merged code. When the full merged code retains logical qubits and does not reduce the code distance in either Pauli sector, the complete lift preserves the amplified memory's certified code-distance lower bounds. A selective lift targets a chosen amplifier label, but its distance must be assessed separately. Fault tolerance of the full measurement protocol requires a separate circuit analysis. Explicit constructions and proofs are given in Methods~\ref{subsec:addressability-methods} and Appendix~\ref{app:logical-addressability}.

\section*{Discussion}

In this work, we introduced distance amplifiers. We established Theorem~\ref{prop:common-joint-amplification} provides a modular route to circuit-distance amplification. Under the stated extraction conditions, a base-code fault-response certificate and a transferable amplifier certificate multiply to give a certificate for the amplified circuit. The resulting certificate can be transferred through subsequent amplification steps, yielding recursive lower bounds on circuit distance. A complementary upper bound from lifting undetected base-code histories gives exact multiplication when the amplifier certificate reaches its code distance and the base certificate is attained. We verify the required amplifier certificates and extraction conditions for three families: flagged $\code{4,2,2}$, canonical hypergraph-product, and rotated surface amplifiers.

This modularity allows different amplifiers to be combined when each stage satisfies the required circuit conditions. The complete preceding code and extraction circuit, including older side data and flags, become the next base, and the sector-wise certificate gains multiply. The inherited logical coordinates also retain an explicit connection between the logical qubits at successive levels. Together with the lifted symmetry and measurement maps, they provide a starting point for developing logical operations on amplified memories.

The results of this paper depend on the specified check presentations, local gate orders, and fault-response assumptions, with ideal encoded input and perfect final checks. Extending the analysis to initialization and logical measurement protocols would connect the memory-distance results to a broader set of fault-tolerant operations. In particular, the algebraic surgery maps require a separate analysis of faults during the complete merge--split procedure. Fixed-depth amplification preserves the LDPC property, while the growth of check weights and qubit degrees under unrestricted recursion motivates further work on controlling these resources.

Distance amplifiers have a simple physical interpretation in terms of the coupled-layer construction~\cite{zhang2026coupledlayerconstructionquantumproduct,ZhangWeiTantivasadakarn2026}. This perspective may provide a route to extending circuit-distance amplification beyond the CSS setting of the present chain-complex formulation. An important next step is to design decoders around the product structure, combining information within base-code copies with the constraints imposed by the coupling checks. The recursive organization may also support decoding strategies that process errors at successive amplification levels. We leave these to future work.

\clearpage

\section*{Methods}
\setcounter{subsection}{0}
\renewcommand{\thesection}{\Alph{section}}
\renewcommand{\theHsection}{methods.\Alph{section}}

\subsection{Recursive code-distance amplification as tensor products}
\label{subsec:code-distance-amplifier}

We derive the code parameters by representing the base and amplifier as chain complexes and retaining the central part of their tensor product. The base and amplifier are CSS codes as length-two chain complexes
\begin{equation}
    \QQQ:Q_2\xrightarrow{H_Z^\intercal}Q_1\xrightarrow{H_X}Q_0,
    \qquad
    \AAA:A_2\xrightarrow{G_Z^\intercal}A_1\xrightarrow{G_X}A_0,
\end{equation}
with $H_XH_Z^\intercal=0$ and $G_XG_Z^\intercal=0$. Data, $X$ checks, and $Z$ checks are indexed by the bases of degrees one, zero, and two. We use $n=\dim Q_1$, $m_X=\dim Q_0$, $m_Z=\dim Q_2$, and the corresponding amplifier counts $n_A,m_X^A,m_Z^A$. Throughout, $k,k_A>0$, and
\begin{equation}
    \operatorname{rank}G_X=m_X^A,
    \qquad \operatorname{rank}G_Z=m_Z^A.
    \label{eq:amplifier-full-rank}
\end{equation}
Thus, $H_0(\AAA)=H_2(\AAA)=0$ and $k_A=n_A-m_X^A-m_Z^A$.

The total complex has $T_r=\bigoplus_{i+j=r}Q_i\otimes A_j$ and differential $\partial_{\QQQ}\otimes I+I\otimes\partial_{\AAA}$. We retain its central terms $T_3\to T_2\to T_1$ to define the CSS code $\TTT$~\cite{AudouxCouvreur2019}. The shaded region below identifies this truncation.

\begin{equation}
\resizebox{\dimexpr\columnwidth-2.5em\relax}{!}{%
\begin{tikzpicture}[
baseline=(T2.base),
font=\small,
term/.style={inner xsep=1pt,inner ysep=2pt},
arr/.style={commutative diagrams/every arrow,shorten <=1pt,shorten >=1pt},
map/.style={font=\scriptsize,anchor=center,inner xsep=1.5pt,inner ysep=1pt,fill=productgreen,text=black},
outermap/.style={map,fill=white},
tmap/.style={font=\scriptsize,inner sep=1pt}
]
\node[term] (T0) at (0,2.70) {$T_0$};
\node[term] (T1) at (0,1.30) {$T_1$};
\node[term] (T2) at (0,0.00) {$T_2$};
\node[term] (T3) at (0,-1.30) {$T_3$};
\node[term] (T4) at (0,-2.70) {$T_4$};
\draw[arr] (T1) -- (T0);
\draw[arr] (T2) -- node[tmap,left=2pt] (t12) {$\widetilde H_X$} (T1);
\draw[arr] (T3) -- node[tmap,left=2pt] (t23) {$\widetilde H_Z^\intercal$} (T2);
\draw[arr] (T4) -- (T3);
\node[term] (Q0A0) at (3.50,2.70) {$Q_0\otimes A_0$};
\node[term] (Q1A0) at (1.20,1.30) {$Q_1\otimes A_0$};
\node[term] (Q0A1) at (5.80,1.30) {$Q_0\otimes A_1$};
\node[term] (Q2A0) at (1.20,0.00) {$Q_2\otimes A_0$};
\node[term] (Q1A1) at (3.50,0.00) {$Q_1\otimes A_1$};
\node[term] (Q0A2) at (5.80,0.00) {$Q_0\otimes A_2$};
\node[term] (Q2A1) at (1.20,-1.30) {$Q_2\otimes A_1$};
\node[term] (Q1A2) at (5.80,-1.30) {$Q_1\otimes A_2$};
\node[term] (Q2A2) at (3.50,-2.70) {$Q_2\otimes A_2$};
\draw[arr] (Q1A0) -- node[outermap] {$H_X\otimes I$} (Q0A0);
\draw[arr] (Q0A1) -- node[outermap] {$I\otimes G_X$} (Q0A0);
\draw[arr] (Q2A0) -- node[map] (mLup) {$H_Z^\intercal\otimes I$} (Q1A0);
\draw[arr] (Q1A1) -- node[map] (mCLup) {$I\otimes G_X$} (Q1A0);
\draw[arr] (Q1A1) -- node[map] (mCRup) {$H_X\otimes I$} (Q0A1);
\draw[arr] (Q0A2) -- node[map] (mRup) {$I\otimes G_Z^\intercal$} (Q0A1);
\draw[arr] (Q2A1) -- node[map] (mLdn) {$I\otimes G_X$} (Q2A0);
\draw[arr] (Q2A1) -- node[map] (mCLdn) {$H_Z^\intercal\otimes I$} (Q1A1);
\draw[arr] (Q1A2) -- node[map] (mCRdn) {$I\otimes G_Z^\intercal$} (Q1A1);
\draw[arr] (Q1A2) -- node[map] (mRdn) {$H_X\otimes I$} (Q0A2);
\draw[arr] (Q2A2) -- node[outermap] {$I\otimes G_Z^\intercal$} (Q2A1);
\draw[arr] (Q2A2) -- node[outermap] {$H_Z^\intercal\otimes I$} (Q1A2);
\begin{scope}[on background layer]
\node[
fit=(T1)(T2)(T3)(t12)(t23)(Q1A0)(Q0A1)(Q2A0)(Q1A1)(Q0A2)(Q2A1)(Q1A2)(mLup)(mCLup)(mCRup)(mRup)(mLdn)(mCLdn)(mCRdn)(mRdn),
inner xsep=5pt,inner ysep=5pt,rounded corners=5pt,fill=productgreen,draw=none
] {};
\end{scope}
\end{tikzpicture}
}
\label{eq:tensor-product-complex}
\end{equation}

The data-qubit space is
\begin{equation}
    T_2=(Q_1\otimes A_1)\oplus(Q_2\otimes A_0)\oplus(Q_0\otimes A_2).
    \label{eq:tensor-qubit-sectors}
\end{equation}
In addition to the central register ($Q_1\ox A_1$), the amplified code contains two side registers. Using the ordering in Eq.~\eqref{eq:tensor-qubit-sectors}, the product check matrices are
\begin{align}
    \widetilde H_X
    &=\begin{pmatrix}
    H_X\otimes I & 0 & I\otimes G_Z^\intercal \\
    I\otimes G_X & H_Z^\intercal\otimes I & 0
    \end{pmatrix},
    \label{eq:tensor-HX}\\
    \widetilde H_Z
    &=\begin{pmatrix}
    H_Z\otimes I & I\otimes G_X^\intercal & 0 \\
    I\otimes G_Z & 0 & H_X^\intercal\otimes I
    \end{pmatrix}.
    \label{eq:tensor-HZ}
\end{align}
The product checks satisfy the CSS condition $\widetilde H_X\widetilde H_Z^\intercal=0$, as follows from the two factor CSS conditions. Counting the data qubits in these three registers gives
\begin{equation}
    \widetilde n=nn_A+m_Zm_X^A+m_Xm_Z^A.
    \label{eq:tensor-length}
\end{equation}

The K\"unneth formula identifies $H_2(\mathcal T)\cong H_1(\QQQ)\otimes H_1(\AAA)$, and the corresponding cohomology (Appendix~\ref{app:distance-amplification}). Hence, $\widetilde k=kk_A$. The $(\alpha_X,\alpha_Z)$-amplifier definition follows Eq.~\eqref{eq:main-static-bounds}, with $\alpha_X,\alpha_Z\geq1$. For dual logical bases, if every $P$-type basis class has at least $\Sigma_P$ representatives with maximum coordinate overlap $o_P$, then $\alpha_X=\Sigma_Z/o_Z$ and $\alpha_Z=\Sigma_X/o_X$, which is proved in Lemma~\ref{lem:distance-amplification}. The pairing of opposite-type logical operators explains this exchange of labels.

For sparsity bounds, let $w(\mathcal C)$ denote the maximum row or column weight of either check matrix in the specified presentation. The block structure gives
\begin{equation}
    w(\QQQ\otimes\AAA)\leq w(\QQQ)+w(\AAA).
    \label{eq:tensor-sparsity}
\end{equation}
Define $\QQQ^{(0)}=\QQQ$ and $\QQQ^{(t+1)}=\QQQ^{(t)}\otimes\AAA$, retaining and reindexing the central three-term presentation at each step. Then
\begin{align}
    k^{(t)}&=kk_A^t,\\
    \left\lceil\alpha_P^t d_P\right\rceil
    &\leq d_P^{(t)}\leq(d_P^A)^t d_P,\\
    w(\QQQ^{(t)})&\leq w(\QQQ)+t w(\AAA).
    \label{eq:recursive-distance-bounds}
\end{align}

\subsection{Data-qubit overhead and amplifier benchmarks}
\label{subsec:amplifier-benchmark}

We now count the additional data qubits under repeated amplification and compare the resulting distance gains.

\begin{table*}[t]
\centering
\caption{\textbf{Selected distance amplifiers and their benchmarks.}
Here $w_{\mathrm{check}}$ denotes the maximum weight of the chosen stabilizer generators (the maximum row weight of either check matrix). The one-step overhead $\eta_A$ and score $\beta_A^{(1)}$ are evaluated for a base-code presentation with $m_X=m_Z=n/2$, corresponding to $\rho_Q=1$. Vertical separators group amplifiers by $\lceil\alpha_A\rceil$. For the even-distance rotated surface codes, the independent check counts are asymmetric: the $\code{4,1,2}$ amplifier has $2$ independent $X$ checks and $1$ independent $Z$ check, while the $\code{16,1,4}$ amplifier has $8$ independent $X$ checks and $7$ independent $Z$ checks.}
\label{tab:balanced-amplifiers}
\resizebox{\textwidth}{!}{%
\begin{tabular}{c*{2}{c}|*{3}{c}|*{2}{c}|c}
\toprule
Amplifier code & Four-qubit & Rotated surface & Steane & Rotated surface & Clustered cyclic & Rotated surface & Bivariate bicycle & Rotated surface \\
$\AAA$ & \cite{Knill2005,Koh2026PhantomCodes} & \cite{GoogleThreshold} & \cite{Steane1996} & \cite{GoogleThreshold} & \cite{QGPU} & \cite{GoogleThreshold} & \cite{liang2025generalized,Wang_2026} & \cite{GoogleThreshold} \\
\midrule
$\code{n_A,k_A,d_A}$ & $\code{4,2,2}$ & $\code{4,1,2}$ & $\code{7,1,3}$ & $\code{9,1,3}$ & $\code{12,4,3}$ & $\code{16,1,4}$ & $\code{18,4,4}$ & $\code{25,1,5}$ \\
$\alpha_A$ & $2$ & $2$ & $7/3$ & $3$ & $3$ & $4$ & $4$ & $5$ \\
$w_{\mathrm{check}}$ & $4$ & $4$ & $4$ & $4$ & $6$ & $4$ & $6$ & $4$ \\
$\eta_A$ & $5$ & $5.5$ & $10$ & $13$ & $16$ & $23.5$ & $25$ & $37$ \\
$\beta_A^{(1)}$ & $\mathbf{0.4307}$ & $0.4066$ & $0.3680$ & $\mathbf{0.4283}$ & $0.3962$ & $\mathbf{0.4391}$ & $0.4307$ & $\mathbf{0.4457}$ \\
$\beta_A^{(\infty)}$ & $\mathbf{0.4104}$ & $0.3869$ & $0.3502$ & $\mathbf{0.4092}$ & $0.3826$ & $\mathbf{0.4226}$ & $0.4165$ & $\mathbf{0.4307}$ \\
\bottomrule
\end{tabular}
}
\end{table*}

Our benchmark compares certified code-distance amplification against total data-qubit overhead. Syndrome auxiliary systems, flag auxiliary systems, gate count, circuit depth, and decoder performance are not included. Temporarily write $x_A=m_X^A$ and $z_A=m_Z^A$, allowing the amplifier to have unequal numbers of $X$- and $Z$-check rows. Under recursive amplification, the retained presentation counts satisfy
\begin{equation}
    \begin{pmatrix} n^{(t+1)}\\m_X^{(t+1)}\\m_Z^{(t+1)}\end{pmatrix}
    =\begin{pmatrix}n_A&z_A&x_A\\x_A&n_A&0\\z_A&0&n_A\end{pmatrix}
    \begin{pmatrix}n^{(t)}\\m_X^{(t)}\\m_Z^{(t)}\end{pmatrix}.
    \label{eq:benchmark-full-recurrence}
\end{equation}
Here, $m_X^{(t)}$ and $m_Z^{(t)}$ denote the numbers of check rows retained in the chosen presentation, rather than the corresponding check ranks.

For a single amplification step, the data-qubit overhead is therefore
\begin{equation}
    \eta_A(\QQQ)=\frac{\widetilde n}{n}
    =n_A+\frac{x_A m_Z+z_A m_X}{n}.
    \label{eq:benchmark-single-step}
\end{equation}
For the reference base-code presentation used in Table~\ref{tab:balanced-amplifiers}, where $m_X=m_Z=n/2$, this reduces to
\begin{equation}
    \eta_A=n_A+\frac{x_A+z_A}{2},
\end{equation}
including the asymmetric case $x_A\neq z_A$. For amplifiers with equal sector-wise guarantees $\alpha_X=\alpha_Z=\alpha_A>1$, we define the one-step score
\begin{equation}
    \beta_A^{(1)}(\QQQ)=\frac{\log\alpha_A}{\log\eta_A(\QQQ)}.
    \label{eq:benchmark-one-time-score}
\end{equation}
For a prescribed finite distance target, however, the more relevant comparison is obtained directly from the exact count recurrence together with the rounded certified distance bound.

Repeated amplification is controlled asymptotically by the dominant eigenvalue of the count matrix. For $x_A,z_A>0$, this is
\begin{equation}
    c_A=n_A+\sqrt{2x_Az_A}.
    \label{eq:benchmark-leading-eigenvalue}
\end{equation}
Any nonnegative initial count vector with $n>0$ has a nonzero projection onto the positive leading eigen-direction, and hence $n^{(t)}=\Theta(c_A^t)$. Combining this with the recursive distance guarantee gives
\begin{equation}
    \beta_A^{(\infty)}=\frac{\log\alpha_A}{\log c_A},
    \qquad
    d^{(t)}=\Omega\left((n^{(t)})^{\beta_A^{(\infty)}}\right).
    \label{eq:benchmark-growth-exponent}
\end{equation}

When the amplifier has equal check counts, $x_A=z_A=m_A$, the recurrence closes on $s^{(t)}=m_X^{(t)}+m_Z^{(t)}$:
\begin{equation}
    \begin{pmatrix}n^{(t+1)}\\s^{(t+1)}\end{pmatrix}
    =\begin{pmatrix}n_A&m_A\\2m_A&n_A\end{pmatrix}
    \begin{pmatrix}n^{(t)}\\s^{(t)}\end{pmatrix}.
    \label{eq:benchmark-count-recurrence}
\end{equation}
In this symmetric case, $c_A=n_A+\sqrt2\,m_A,\qquad \eta_A=n_A+m_A\rho_Q \tand \rho_Q=\frac{m_X+m_Z}{n}$.

\subsection{Recursive staged joint syndrome extraction}
\label{subsec:syndrome-circ-amplified}

In this subsection, we specify the gate schedule used in the circuit-distance analysis, including when the new side data are coupled and when the check auxiliary systems are read. As shown in Figure~\ref{fig:css_tensor_overview}, the amplified-code data registers are $Q_1\otimes A_1$, $Q_2\otimes A_0$, and $Q_0\otimes A_2$, ordered as in Eq.~\eqref{eq:tensor-qubit-sectors}. The $X$-syndrome families are indexed by $Q_0\otimes A_1$ and $Q_1\otimes A_0$, while the $Z$-syndrome families are indexed by $Q_2\otimes A_1$ and $Q_1\otimes A_2$. An $X$-syndrome ancilla is prepared in $|+\rangle$, acts as the control for its data CNOTs, and is measured in the $X$ basis; a $Z$-syndrome ancilla is prepared in $|0\rangle$, acts as the target of its data CNOTs, and is measured in the $Z$ basis.

For each base-code or amplifier check, we retain its original measurement circuit but defer the final syndrome-auxiliary readout. We refer to this retained subcircuit as its \emph{open core}. The open core contains all inherited data couplings and flag operations, including flag measurements. The subsequent couplings of the same syndrome ancilla to the newly introduced side data form the gadget's \emph{side tail}; the ancilla is measured at the end of the extraction round.

The SJ circuit is organized into four ordered coupling windows,
\begin{equation}
    \mathrm{base\ code}\ \longrightarrow\ \mathrm{amplifier} \longrightarrow\ \mathrm{side\ 1}\ \longrightarrow\ \mathrm{side\ 2},
\label{eq:tensor_parallel_windows}
\end{equation}
with the corresponding parity-check blocks shown in Figure~\ref{fig:css_tensor_overview}~(b). Within each window, gates may be parallelized provided that the local coupling order of every check is preserved and no qubit participates in more than one gate per layer. The base-code window consists of $n_A$ disjoint copies of the base-code schedule, while the amplifier window consists of $n$ disjoint copies of the amplifier schedule. The two side-coupling windows complete the product checks through disjoint $X$- and $Z$-coupling families. All syndrome auxiliary systems are read out only after side~2, while any inherited local flags are measured before the newly appended side tail of the corresponding gadget.

The construction is recursive. After one amplification step, the entire amplified code and its SJ extraction circuit become the base code and schedule for the next step. Thus $Q_1$ at the next level includes all previously introduced side data, and the inherited base-code core contains all couplings accumulated in earlier amplification steps.

At every level, the same four-window ordering is retained. Each inherited constituent core preserves its original internal $X/Z$ interleaving and local coupling order; recursion only defers its terminal syndrome readout and appends the new side couplings. In this sense, the construction reuses the complete current extraction circuit rather than rescheduling it from scratch.

\subsection{Circuit-distance amplification}
\label{subsec:circuit-distance-methods}

Syndrome-auxiliary faults can propagate to several data qubits. We analyze these correlations in the staged joint (SJ) construction using a fault-response certificate $D_P^Q$ for the base-code circuit $\mathsf S_Q$ and a transferable amplifier certificate $D_P^A$.

An \emph{extraction gadget} is the complete circuit measuring one stabilizer check, including its syndrome ancilla and any flags. All circuit-distance comparisons use the same fixed number of noisy extraction rounds, with ideal encoded input and perfect final check measurements. For $P\in\{X,Z\}$, let $\bar P$ denote the opposite Pauli type. For a pure-$P$ fault history $F$, write $E_P(F)$ for its final binary data-error pattern and define
\begin{align}
    \mathsf D(F)&=\bigl(s_{\mathrm{noisy}}(F),f(F),s_{\mathrm{final}}(F)\bigr),
    \label{eq:iva-complete-record}\\
    d_{\mathrm{circ},P}(\mathsf S_Q)
    &=\min_{\substack{F:\,\mathsf D(F)=0,\\{[E_P(F)]\neq0}}}|F|.
    \label{eq:iva-circuit-distance}
\end{align}
The complete record contains every noisy check outcome, every individual flag outcome and the perfect final syndrome, all relative to the fault-free experiment. Each faulty location contributes one to $|F|$, including a two-qubit fault at a single CNOT. Zero final syndrome ensures $E_P(F)\in\ker H_{\bar P}$, while $[E_P(F)]\neq0$ excludes stabilizers in $\operatorname{im}H_P^\intercal$.

The lower-bound proof contracts product logicals to base-code data responses that need not arise from complete undetected histories. We therefore require a fault-response certificate, rather than a circuit-distance lower bound alone.

\begin{definition}[Fault-response certificate]
    \label{def:joint-response-certificate}
    For a fixed extraction circuit, let $e_{\mathsf G}(B)$ and $f_{\mathsf G}(B)$ denote the data and flag responses of an allowed pure-$P$ fault bundle $B$ in a gadget $\mathsf G$. A response decomposition of a data pattern $e$ has the form
    \begin{align}
        e&=\sum_\lambda e_{\mathsf G_\lambda}(B_\lambda)+\sum_{\mu=1}^{s}u_\mu,\\
        K&:=\sum_\lambda|B_\lambda|+s,
    \end{align}
    where each $\mathsf G_\lambda$ measures a $P$ check, $f_{\mathsf G_\lambda}(B_\lambda)=0$, and the $u_\mu$ are weight-one data patterns, with $s\geq0$. An integer $D_P\geq1$ is a fault-response certificate if every such decomposition of every $e\in\ker H_{\bar P}\setminus\operatorname{im}H_P^\intercal$ has assigned fault count $K\geq D_P$.
\end{definition}

Ordinary syndrome outcomes are unrestricted, and repeated gadget labels and singleton terms are allowed. Each flagged bundle is kept intact because its individual faults may have nonzero flag responses that cancel collectively. Every zero-record history admits a decomposition with $K\leq|F|$, giving $D_P\leq d_{\mathrm{circ},P}(\mathsf S)$, where the proof is given in Appendix~\ref{app:circuit_distance_amplification}.

\begin{definition}[Transferable amplifier certificate] \label{def:transferable-amplifier-certificate}
    An integer $D_P^A\geq1$ is a transferable amplifier certificate if every nontrivial amplifier logical $a\in\ker G_{\bar P}\setminus\operatorname{im}G_P^\intercal$ admits $D_P^A$ representatives $\omega_i$ with pairwise disjoint supports satisfying
    \begin{equation}
        G_P\omega_i=0,\qquad \omega_i^\intercal a=1,
    \label{eq:methods-amplifier-detection}
    \end{equation}
    and
    \begin{equation}
        \operatorname{wt}\!\left[
        \bigl(\omega_i^\intercal e_{\mathsf G_A}(B)\bigr)_{i=1}^{D_P^A}
        \right]\leq |B|
    \label{eq:methods-amplifier-parity}
    \end{equation}
    for every allowed zero-flag pure-$P$ bundle $B$ in every amplifier $P$-check gadget $\mathsf G_A$. The representatives may depend on $a$, but the same collection must satisfy the bound for all such gadgets and bundles.
\end{definition}

To apply these certificates, write the central data pattern of a nontrivial product logical as a matrix $M$, with rows indexed by base-code qubits and columns by amplifier qubits. The contraction $M\omega_i$ is the modulo-two sum of the columns selected by $\omega_i$. Under the SJ projection conditions and full row rank of $G_X,G_Z$, the transferable amplifier certificate supplies $D_P^A$ representatives for which every contraction is a nontrivial base-code logical. Disjointness and Eq.~\eqref{eq:methods-amplifier-parity} ensure that a product-response decomposition of assigned fault count $K$ induces base-code decompositions with counts $C_i$ satisfying
\begin{equation}
    \sum_{i=1}^{D_P^A}C_i\leq K.
\end{equation}
Applying the base-code certificate to each contraction gives $K\geq D_P^A D_P^Q$. Because this holds for every logical-response decomposition, it transfers the certificate itself, not only the circuit-distance lower bound.

For the upper bound, we copy a minimum-fault undetected base-code history along a minimum-weight amplifier logical of weight $d_P^A$. The SJ history-lifting conditions preserve the zero complete record and nontrivial logical output, producing an undetected product history with $d_P^A$ times as many faults.

\begin{theorem}[Multiplication of fault-response certificates] \label{prop:common-joint-amplification}
    Suppose $\mathsf S_Q$ has fault-response certificate $D_P^Q$, the amplifier has a transferable amplifier certificate $D_P^A$, and $G_X,G_Z$ have full row rank. If the SJ circuit satisfies the projection and history-lifting conditions of Appendix~\ref{app:circuit_distance_amplification}, then
    \begin{equation}
        D_P^A D_P^Q
        \leq d_{\mathrm{circ},P}(\widetilde{\mathsf S})
        \leq d_P^A\,d_{\mathrm{circ},P}(\mathsf S_Q).
        \label{eq:common-joint-bounds}
    \end{equation}
    The amplified circuit inherits the certificate $D_P^A D_P^Q$. If $D_P^A=d_P^A$ and $D_P^Q=d_{\mathrm{circ},P}(\mathsf S_Q)$, then
    \begin{equation}
        d_{\mathrm{circ},P}(\widetilde{\mathsf S})
        =d_P^A\,d_{\mathrm{circ},P}(\mathsf S_Q),
        \label{eq:common-joint-equality}
    \end{equation}
    and the inherited certificate is attained.
\end{theorem}

The full-row-rank requirement applies only to the amplifier; the base-code checks may be dependent. The specified flagged $\code{4,2,2}$, canonical HGP with full-row-rank seed matrices, and rotated-surface extraction circuits satisfy $D_P^A=d_P^A$.

For a sequence of amplifiers $\AAA_1,\ldots,\AAA_t$, preserving the required SJ interface at every stage gives
\begin{equation}
    \left(\prod_{j=1}^{t}D_P^{\AAA_j}\right)D_P^{(0)}
    \leq d_{\mathrm{circ},P}(\mathsf S^{(t)})
    \leq
    \left(\prod_{j=1}^{t}d_P^{\AAA_j}\right)
    d_{\mathrm{circ},P}(\mathsf S^{(0)}).
    \label{eq:mixed-circuit-bounds}
\end{equation}
Once the amplifier certificates and interface conditions are established, only the initial base-code certificate requires an independent verification. Subsequent certificates follow by transfer.

Writing $d_{\mathrm{circ}}=\min_{P\in\{X,Z\}}d_{\mathrm{circ},P}$, a fixed symmetric amplifier with $D_X^A=D_Z^A=d_X^A=d_Z^A=d_A$ gives exact scalar recursion whenever both seed sectors have a certificate at the common value $D_0=d_{\mathrm{circ}}(\mathsf S^{(0)})$:
\begin{equation}
    d_{\mathrm{circ}}(\mathsf S^{(t)})=d_A^tD_0.
    \label{eq:common-scalar-recursion}
\end{equation}
Separate attainment in both seed sectors is not required.

\subsection{Logical coordinates and measurement maps}
\label{subsec:addressability-methods}

The inherited labels $(\mu,a)$ in Eq.~\eqref{eq:canonical-basis-product} identify the logical qubits of the amplified code. Here we specify how compatible base-code permutations act on these coordinates and how surgery maps impose logical constraints at all amplifier labels or at a selected label. We retain the chain conventions and full-row-rank assumption on the amplifier check matrices of Methods~\ref{subsec:code-distance-amplifier}. We write $\mathcal T$ for the full tensor complex underlying $\TTT$ and construct surgery maps before taking the central CSS truncation.

A base data permutation $\Pi_1$ is compatible with the displayed check presentation when check-row permutations $\Pi_0$ and $\Pi_2$ satisfy
\begin{equation}
    H_X\Pi_1=\Pi_0H_X,
    \qquad H_Z^{\intercal}\Pi_2=\Pi_1H_Z^{\intercal}.
    \label{eq:permutation-chain-condition}
\end{equation}
In the sector ordering of Eq.~\eqref{eq:tensor-qubit-sectors}, its lifted data permutation is
\begin{equation}
    \widetilde\Pi_2=\operatorname{diag}\bigl(
    \Pi_1\otimes I_{A_1},\,
    \Pi_2\otimes I_{A_0},\,
    \Pi_0\otimes I_{A_2}\bigr).
    \label{eq:lifted-data-permutation}
\end{equation}
This induces the same base logical action at every amplifier label. The check-row condition is essential because base checks index physical side data; preservation of check row spaces alone does not guarantee such a physical-permutation lift.

To specify logical measurements, we begin with a base surgery and lift the map defining its logical constraints. For a $Z$-type homological measurement, we use the mapping-cone construction of Ref.~\cite{IdeGowda2025} in our chain convention. Let
\begin{equation}
    \mathcal{U} :U_1\xrightarrow{\partial^U_1}U_0\xrightarrow{\partial^U_0}U_{-1}
\end{equation}
be an auxiliary complex and let $f:\mathcal U\to\QQQ$ be a chain map. Its potentially nonzero components are $f_1:U_1\to Q_1$ and $f_0:U_0\to Q_0$, satisfying $H_Xf_1=f_0\partial^U_1$; all other components vanish. The merged complex is
\begin{equation}
    \begin{gathered}
    \mathcal M=\operatorname{Cone}(f),\qquad M_r=Q_r\oplus U_{r-1},\\
    \partial^M_r=\begin{pmatrix}\partial^Q_r&f_{r-1}\\0&\partial^U_{r-1}\end{pmatrix}.
    \end{gathered}
    \label{eq:surgery-cone}
\end{equation}
The base logical classes promoted to merged-code stabilizers form the subspace
\begin{equation}
    \mathcal W=\operatorname{im}\bigl(f_*:H_1(\mathcal U)\to H_1(\QQQ)\bigr).
    \label{eq:surgery-promoted-space}
\end{equation}
Set $m=\dim\mathcal W$ and $\kappa=\dim\ker\bigl(f_{*,0}:H_0(\mathcal U)\to H_0(\QQQ)\bigr)$. The merged code has $k-m+\kappa$ logical qubits: $m$ base logical degrees of freedom are removed, while the auxiliary complex contributes $\kappa$ additional ones. These must be retained explicitly or treated through a specified gauge-fixing prescription.

To impose the base constraints at every amplifier label, the complete lift tensors the entire map, $F=f\otimes I_{\AAA}:\mathcal U\otimes\AAA\to\mathcal T$. Regrouping tensor summands gives
\begin{equation}
    \operatorname{Cone}(f\otimes I_{\AAA})
    \cong\operatorname{Cone}(f)\otimes\AAA.
    \label{eq:cone-tensor-identity}
\end{equation}
This chain isomorphism identifies the degree-two merged code with the amplification of the base merged code. It promotes $\mathcal W\otimes H_1(\AAA)$ and has $(k-m+\kappa)k_A$ logical qubits. All components of the chain map, including the couplings outside the central register, are required for this identity.

To select one amplifier logical direction, choose $\gamma\in\ker G_X$ with $[\gamma]\ne0$ in $H_1(\AAA)$ and define the shifted auxiliary complex by $\mathcal U[1]_r=U_{r-1}$. The map
\begin{equation}
    F^\gamma_r:\mathcal U[1]_r\longrightarrow T_r,
    \qquad u\longmapsto f_{r-1}(u)\otimes\gamma
    \label{eq:selective-surgery-map}
\end{equation}
is a chain map. The degree-two truncation of $\operatorname{Cone}(F^\gamma)$ promotes precisely $\mathcal W\otimes\operatorname{span}\{[\gamma]\}$ and has $kk_A-m+\kappa$ logical qubits. Taking $\gamma=\ell^A_{Z,a}$ selects the inherited coordinate $a$. In particular, choosing $\mathcal W=\operatorname{span}\{[\ell_{Z,\mu}]\}$ targets the logical operator $\bar Z_{\mu,a}$; choosing a base logical parity in $\mathcal W$ targets that parity at the same amplifier label.

The selective map also gives an explicit stabilizer presentation. Embed $F^\gamma_2=f_1\otimes\gamma$ into the central summand of $T_2$, and $F^\gamma_1=f_0\otimes\gamma$ into $Q_0\otimes A_1\subseteq T_1$. Using the product checks in Eqs.~\eqref{eq:tensor-HX} and~\eqref{eq:tensor-HZ}, the merged code on data $T_2\oplus U_0$ has checks
\begin{equation}
    \begin{aligned}
    H_X^\gamma&=\begin{pmatrix}\widetilde H_X&F^\gamma_1\\0&\partial^U_0\end{pmatrix},\\
    H_Z^\gamma&=\begin{pmatrix}\widetilde H_Z&0\\(F^\gamma_2)^{\intercal}&(\partial^U_1)^{\intercal}\end{pmatrix}.
    \end{aligned}
    \label{eq:selective-surgery-checks}
\end{equation}
CSS commutativity follows from $\widetilde H_XF^\gamma_2=F^\gamma_1\partial^U_1$ and $\partial^U_0\partial^U_1=0$. Separate auxiliary copies with images $\mathcal W_a$ and independent classes $[\gamma_a]$ promote $\bigoplus_a\mathcal W_a\otimes\operatorname{span}\{[\gamma_a]\}$.

The support of the selected amplifier representative controls the connection cost. Writing $|\gamma|$ for Hamming weight and $\operatorname{cw}$ and $\operatorname{rw}$ for maximum column and row weights,
\begin{equation}
    \begin{aligned}
        \operatorname{cw}(f_i\otimes\gamma)&=|\gamma|\operatorname{cw}(f_i),\\
        \operatorname{rw}(f_i\otimes\gamma)&=\operatorname{rw}(f_i).
    \end{aligned}
    \label{eq:selective-connection-weights}
\end{equation}
Thus, selection can increase connection-check weights or qubit degrees, depending on the block. The multiplicative overhead from $\gamma$ is bounded for a fixed amplifier, but can grow when $\gamma$ is itself a recursively amplified logical representative.

These constructions specify algebraic logical access and can be reapplied to the complete current code at each amplification step. Appendix~\ref{app:logical-addressability} proves the lifts and states the distance guarantees for the merged codes; fault tolerance of the full measurement protocol additionally requires an analysis of preparation and merge--split boundaries.

\section*{Acknowledgements}

We are grateful to {Nathanan Tantivasadakarn} for valuable discussions. Z.W. thanks {Christopher A. Pattison} for helpful feedback on the manuscript. {Y.-A.C.} and {Z.W.} are supported by the National Natural Science Foundation of China (Grant No.~12474491) and the Central University Fundamental Research Funds (Peking University). The Berlin team has been supported by the BMFTR (QSolid, MUNIQC-Atoms, PasQuops), the DFG (CRC 183, BoLaCo, and SPP 2514), the Quantum Flagship (Millenion, PasQuanS2), the Munich Quantum Valley, Berlin Quantum, QuantERA (SDPCode), and the European Research Council (DebuQC).

\section*{Author contributions}

Z.W. conceived this project. Z.L., B.G., and Z.W. carried out the preliminary simulations and proved circuit distance bounds with inspiring guidance from Y.-A.C. and J.E. Z.L. devised the staged joint syndrome extraction circuit and carried out the main simulations presented in the manuscript. B.G. came up with the notion of certificate to unify circuit distance proofs and proved multiplication of fault-response certificates. All authors contributed extensively to writing the manuscript.

\section*{Data availability}

Code and data for the circuit-level simulations and the computation of the fault-response certificates $D_P$ are available at \url{https://github.com/zijian-liang/dp_certificates_github}.

\bibliographystyle{apsrev4-2}
\bibliography{references}

\appendix
\renewcommand{\theHsection}{appendix.\Alph{section}}
\clearpage

\section{Review of K\"unneth formula and code-distance amplification}
\label{app:distance-amplification}

In this appendix, we review the basic notions behind the tensor construction of CSS product codes and the code-distance multiplication results from Ref.~\cite{AudouxCouvreur2019} for tensoring two CSS codes. We first recall the K\"unneth formula and its interpretation for CSS tensor products. All chain complexes considered below are finite-dimensional over $\mathbb{F}_2$, with fixed coordinate bases. For a chain complex $\mathcal{C}$ with differentials $\partial_r^{\mathcal{C}}:C_r\rightarrow C_{r-1}$, its homology and cohomology are
\begin{align}
    H_r(\mathcal{C})&=\ker\partial_r^{\mathcal{C}}\big/\operatorname{im}\partial_{r+1}^{\mathcal{C}},\\
    H^r(\mathcal{C})&=\ker\bigl(\partial_{r+1}^{\mathcal{C}}\bigr)^\intercal\big/\operatorname{im}\bigl(\partial_r^{\mathcal{C}}\bigr)^\intercal.
\end{align}
For two chain complexes $\QQQ$ and $\AAA$, their total tensor-product complex $\mathcal{T}$ has chain groups and differential
\begin{equation}
    T_r=\bigoplus_{i+j=r}Q_i\otimes A_j
    \tand
    \partial_{\mathcal{T}}=\partial_{\QQQ}\otimes I+I\otimes\partial_{\AAA}.
\end{equation}
Since the coefficient ring is a field, the K\"unneth formula contains no torsion terms and gives
\begin{align}
    H_r(\mathcal{T})&\cong\bigoplus_{i+j=r}H_i(\QQQ)\otimes H_j(\AAA),
    \label{eq:app-kunneth-general-homology}\\
    H^r(\mathcal{T})&\cong\bigoplus_{i+j=r}H^i(\QQQ)\otimes H^j(\AAA).
    \label{eq:app-kunneth-general-cohomology}
\end{align}
These isomorphisms are induced by tensoring representatives. For cycles $u$ and $v$, the tensor $[u]\otimes[v]$ is mapped to $[u\otimes v]$, and similarly for cocycles. In particular, tensor products of homology bases give a homology basis of the total complex.

We apply this construction to the CSS complexes of the base code and the amplifier,
\begin{align}
    \QQQ:&\quad Q_2\xrightarrow{H_Z^\intercal}Q_1\xrightarrow{H_X}Q_0,\\
    \AAA:&\quad A_2\xrightarrow{G_Z^\intercal}A_1\xrightarrow{G_X}A_0,
\end{align}
with parameters $\code{n,k,d_X,d_Z}$ and $\code{n_A,k_A,d_X^A,d_Z^A}$, respectively. The $Z$- and $X$-logical spaces of the base code are
\begin{equation}
    H_1(\QQQ)=\ker H_X/\operatorname{im}H_Z^\intercal,\qquad
    H^1(\QQQ)=\ker H_Z/\operatorname{im}H_X^\intercal,
\end{equation}
and analogously for $\AAA$. The code distances are the minimum Hamming weights of representatives of nonzero classes in these spaces, respectively.

The CSS tensor-product code $\TTT=\QQQ\otimes\AAA$ is defined by the central truncation introduced in Eq.~\eqref{eq:tensor-product-complex}
\begin{equation}
    T_3\xrightarrow{\widetilde H_Z^\intercal}T_2\xrightarrow{\widetilde H_X}T_1,
    \label{eq:app-central-truncation}
\end{equation}
where
\begin{equation}
    T_2=(Q_1\otimes A_1)\oplus(Q_2\otimes A_0)\oplus(Q_0\otimes A_2).
    \label{eq:app-product-qubit-space}
\end{equation}
Thus, its $Z$- and $X$-logical spaces are $H_2(\mathcal{T})$ and $H^2(\mathcal{T})$, respectively. The central truncation does not change either of these groups, since they depend only on the maps adjacent to $T_2$. In general, Eq.~\eqref{eq:app-kunneth-general-homology} gives
\begin{align}
    H_2(\mathcal{T})\cong{}&
    \bigl(H_2(\QQQ)\otimes H_0(\AAA)\bigr)
    \nonumber\oplus
    \bigl(H_1(\QQQ)\otimes H_1(\AAA)\bigr)
    \nonumber\\
    &\oplus\bigl(H_0(\QQQ)\otimes H_2(\AAA)\bigr).
    \label{eq:app-kunneth-all-sectors}
\end{align}
As $G_X$ and $G_Z$ have full row rank, namely, $H_0(\AAA)=H_2(\AAA)=0$, the K\"unneth decompositions in Eqs.~\eqref{eq:app-kunneth-general-homology} and~\eqref{eq:app-kunneth-general-cohomology} reduce to
\begin{align}
    H_2(\mathcal{T})&\cong H_1(\QQQ)\otimes H_1(\AAA),
    \label{eq:app-kunneth-homology-group}\\
    H^2(\mathcal{T})&\cong H^1(\QQQ)\otimes H^1(\AAA).
    \label{eq:app-kunneth-cohomology-group}
\end{align}
So, $\TTT$ encodes $\widetilde k=kk_A$ logical qubits, and every logical class admits a representative supported entirely in the space $Q_1\otimes A_1$.

A lower bound on the product-code distance can be obtained by contracting a product-code logical operator against multiple representatives of an amplifier logical class. The relevant quantity is the overlap of these representatives. For a nonempty family $\Omega\subseteq\mathbb{F}_2^{n_A}$, define
\begin{equation}
    \operatorname{overlap}(\Omega):=\max_{a\in\{1,\ldots,n_A\}}
    \bigl|\{\omega\in\Omega:\omega_a=1\}\bigr|.
    \label{eq:app-logical-overlap}
\end{equation}
Thus, the overlap is the maximum number of family members acting on any one physical qubit. The following lemma states the basis-aware version of the overlap lower bound in Lemma~2.7 and Theorem~2.8 of Ref.~\cite{AudouxCouvreur2019}, together with the tensor-product upper bound of Corollary~2.14 of the same reference.

\begin{lemma}[Code-distance amplification] \label{lem:distance-amplification}
    Let $\QQQ$ and $\AAA$ be as above, with $H_0(\AAA)=H_2(\AAA)=0$. Let $\{\ell_{X,i}\}_{i=1}^{k_A}$ and $\{\ell_{Z,i}\}_{i=1}^{k_A}$ be dual $X$- and $Z$-logical operator bases of $\AAA$. For each $P\in\{X,Z\}$ and $i=1,\ldots,k_A$, suppose there exists a family $\Omega_{P,i}$ of physical representatives of the logical class $[\ell_{P,i}]$ such that
    \begin{equation}
        |\Omega_{P,i}|\geq\Sigma_P
        \tand \operatorname{overlap}(\Omega_{P,i})\leq o_P,
    \end{equation}
    where $\Sigma_P\geq o_P>0$ are integers. Then, $\AAA$ is an $(\alpha_X,\alpha_Z)$-distance amplifier with certified factors
    \begin{equation}
        \alpha_X=\frac{\Sigma_Z}{o_Z}\tand\alpha_Z=\frac{\Sigma_X}{o_X}.
        \label{eq:certified-amplification-factors}
    \end{equation}
    More precisely, the code distances of $\TTT=\QQQ\otimes\AAA$ satisfy
    \begin{equation}
        \left\lceil\alpha_Pd_P\right\rceil\leq\widetilde d_P\leq d_Pd_P^A,
        \qquad P\in\{X,Z\},
        \label{eq:amplifier-distance-bounds}
    \end{equation}
    where $d_P$ and $d_P^A$ are the code distances of the base code and amplifier, respectively.
\end{lemma}

\begin{proof}
    We first prove the lower bound on $\widetilde d_Z$. Let $z$ be a nontrivial $Z$-logical representative of $\TTT$, and denote its component in $Q_1\otimes A_1$ by $z_{11}$. Using Eq.~\eqref{eq:app-kunneth-homology-group}, we identify its logical class with
    \begin{equation}
        [z]=\sum_{j=1}^{k_A}[u_j]\otimes[\ell_{Z,j}],
        \qquad [u_j]\in H_1(\QQQ).
    \end{equation}
    Since $[z]\neq 0$, there exists an index $i$ such that $[u_i]\neq 0$. For each $\omega\in\Omega_{X,i}$, define the contraction
    \begin{equation}
        z_\omega:=(I_{Q_1}\otimes\omega^\intercal)z_{11}.
        \label{eq:app-logical-contraction}
    \end{equation}
    Because $\omega$ represents the logical class $[\ell_{X,i}]$, it satisfies $G_Z\omega=0 \tand \omega^\intercal\ell_{Z,j}=\delta_{ij}$. Its induced action on logical classes is therefore determined by the dual pairing
    \begin{equation}
        [z_\omega]=\sum_{j=1}^{k_A}\bigl(\omega^\intercal\ell_{Z,j}\bigr)[u_j]
        =[u_i]\neq 0.
    \end{equation}
    Thus, every $z_\omega$ is a nontrivial $Z$-logical representative of $\QQQ$, and $|z_\omega|\geq d_Z$. To bound the total weight of these contractions, write $z_{11}=\sum_{a=1}^{n_A}z_a\otimes e_a$, where $\{e_a\}$ is the physical-coordinate basis of $A_1$ and $z_a\in Q_1$. Then,
    \begin{equation}
        z_\omega=\sum_{a:\,\omega_a=1}z_a\tand
        |z_{11}|=\sum_{a=1}^{n_A}|z_a|.
    \end{equation}
    Each amplifier coordinate occurs in at most $o_X$ members of $\Omega_{X,i}$. Using the triangle inequality for Hamming weight gives
    \begin{align}
        \Sigma_Xd_Z&\leq\sum_{\omega\in\Omega_{X,i}}|z_\omega|
        \nonumber\\
        &\leq\sum_{a=1}^{n_A}|z_a|\,
        \bigl|\{\omega\in\Omega_{X,i}:\omega_a=1\}\bigr|
        \nonumber\\
        &\leq o_X|z_{11}|\leq o_X|z|.
        \label{eq:app-overlap-weight-bound}
    \end{align}
    Minimizing over all nontrivial $Z$-logical representatives yields
    \begin{equation}
        \widetilde d_Z\geq\left\lceil\frac{\Sigma_X}{o_X}d_Z\right\rceil
        =\left\lceil\alpha_Zd_Z\right\rceil.
    \end{equation}
    Applying the same argument to cohomology, using Eq.~\eqref{eq:app-kunneth-cohomology-group} and the families $\Omega_{Z,i}$, yields
    \begin{equation}
        \widetilde d_X\geq\left\lceil\frac{\Sigma_Z}{o_Z}d_X\right\rceil
        =\left\lceil\alpha_Xd_X\right\rceil.
    \end{equation}
    For the upper bounds, fix $P\in\{X,Z\}$ and choose minimum-weight nontrivial $P$-logical representatives $l$ and $l^A$ of $\QQQ$ and $\AAA$, respectively. The operator represented by $l\otimes l^A$, supported entirely in $Q_1\otimes A_1$, has a nonzero logical class by the corresponding homological or cohomological K\"unneth decomposition. Its weight is
    \begin{equation}
        |l\otimes l^A|=|l|\,|l^A|=d_Pd_P^A,
    \end{equation}
    so $\widetilde d_P\leq d_Pd_P^A$.
\end{proof}

The exchange of $X$ and $Z$ in Eq.~\eqref{eq:certified-amplification-factors} reflects the pairing between opposite-type logical operators: families of $X$-logical representatives certify the $Z$-distance lower bound, and conversely. The certified factors need not equal the amplifier distances. When $\alpha_P=d_P^A$, the lower and upper bounds coincide and establish
\begin{equation}
    \widetilde d_P=d_Pd_P^A.
\end{equation}

\section{Fault response and circuit-distance amplification}
\label{app:circuit_distance_amplification}

We first state how faults in individual extraction gadgets contribute to a final data error, then prove how the resulting certificates transfer through the SJ circuit. Each extraction gadget uses a private syndrome ancilla, controlling the data CNOTs for an $X$ check and receiving them for a $Z$ check. The permitted local flags have zero ideal outcome and are unaffected by incoming data errors. With these directions an incoming data Pauli does not spread to another data qubit: the syndrome component it induces cannot propagate through later data or permitted flag couplings. Pure-$P$ faults in an opposite-type gadget, and any other allowed locations outside same-type gadgets, contribute at most one data-$P$ singleton each. We assume binary-linear fault propagation and ideal-correct extraction. A pure-$P$ fault has only $I$ or $P$ components on the location's qubits, in the corresponding CSS Pauli frame.

We use the response decompositions of Definition~\ref{def:joint-response-certificate}. A syndrome-auxiliary fault may create a multi-qubit suffix at assigned fault count one; flag cancellations must therefore be handled within complete bundles. Incoming errors are excluded when computing each gadget's own response.

\begin{lemma}[Reduction of a complete history] \label{lem:common-history-bounds}
    Every pure-$P$ history $F$ with $\mathsf D(F)=0$ admits a response decomposition of $E_P(F)$ with $K\leq|F|$.
\end{lemma}
\begin{proof}
    Group the faults within each occurrence of a $P$-check gadget. Incoming data errors do not affect its flags, so each resulting bundle has zero flag response. Each remaining fault supplies at most one singleton. Linearity and the absence of subsequent data-to-data propagation make these responses sum to $E_P(F)$. Their underlying faulty locations are disjoint, giving $K\leq|F|$.
\end{proof}

It then follows that $D_P\leq d_{\mathrm{circ},P}(\mathsf S)$. If every qubit of a minimum-weight logical has an allowed data-only fault location after its last noisy interaction, placing these faults in the final round also gives $d_{\mathrm{circ},P}(\mathsf S)\leq d_P(\QQQ)$. A data-only output fault on each qubit's last syndrome--data CNOT is sufficient; full row rank alone does not imply this coverage. The certificate depends on each gadget's local order---its sequence of data and flag couplings and flag readout---and its allowed fault responses. It is unchanged by global interleavings that preserve these responses and the stated propagation assumptions.

In the SJ construction, the central cores precede the newest side tails, and all relevant flags are read before those tails. Older side data remain part of the current base code. Ideal correctness follows by reordering gates with
\begin{equation}
    \operatorname{CX}_{d,z}\operatorname{CX}_{x,d}
    =\operatorname{CX}_{x,d}\operatorname{CX}_{d,z}\operatorname{CX}_{x,z}.
    \label{eq:common-cnot-exchange}
\end{equation}
Same-factor exchanges cancel by validity of the constituent schedules; the prescribed cross-factor order gives either no exchange or two identical canceling ancilla CNOTs. The flags disentangle ideally. The response conditions used below are explicit: discarding only the newest side data by $\pi$, every allowed zero-flag pure-$P$ product bundle $\widehat B$ projects as
\begin{equation}
    \pi e_{\widetilde{\mathsf G}}(\widehat B)=
    \begin{cases}
        e_{\mathsf G_Q}(B)\delta_b^\intercal,
        &\text{extended base check in copy }b,\\
        \delta_q e_{\mathsf G_A}(B)^\intercal,
        &\text{new coupling check at datum }q,
    \end{cases}
    \label{eq:common-gadget-projections}
\end{equation}
where $B$ is an allowed zero-flag bundle of the corresponding constituent gadget with $|B|\leq|\widehat B|$, and $\delta$ denotes a coordinate unit vector. These conditions include faults on newly added gates: deleting newest-tail faults leaves the central response and earlier flags unchanged.

The upper bound additionally uses synchronized copies of the complete base cores, followed by the amplifier cores and then the two side-coupling windows. The side~1 window finishes before the side~2 window samples the same side coordinates. Syndrome readouts are deferred until the required tails finish, and every base fault location embeds with the same allowed fault and count, including faults at deferred readouts. Ideal-correct interleaving within a core is permitted. These timing conditions ensure cancellation before sampling. The projection conditions alone do not imply a complete-history lift.

\subsection{Certificate multiplication and history lifting}
\label{app:common-budgets}

The lower bound contracts each product logical into base-code logical responses; the upper bound lifts an undetected base history. Write a product $P$-error as $E=(M,V,U)$, with $M\in\mathbb F_2^{n_Q\times n_A}$, $V\in\mathbb F_2^{m_{\bar P}^Q\times m_P^A}$, and $U\in\mathbb F_2^{m_P^Q\times m_{\bar P}^A}$. The side-register names exchange with $P$. Its normalizer equations and the generating forms of product $P$ stabilizers are
\begin{align}
    H_{\bar P}M+VG_P&=0,\qquad MG_{\bar P}^\intercal=H_P^\intercal U,
    \label{eq:common-central-normalizer}\\
    (H_P^\intercal B,0,BG_{\bar P}^\intercal),&\qquad
    (CG_P,H_{\bar P}C,0),
    \label{eq:common-central-stabilizers}
\end{align}
where $B,C$ have compatible dimensions. For an opposite-type amplifier logical $\omega$, the \emph{contraction} $M\omega$ sums the central columns on its support and satisfies $H_{\bar P}M\omega=0$. To ensure nontriviality, we first establish the following pairing fact.

\begin{lemma}[Central logical pairing] \label{lem:common-central-pairing}
    If $G_X,G_Z$ have full row rank, every nontrivial product $P$ logical $E=(M,V,U)$ admits $x\in\ker H_P$ for which $a=M^\intercal x$ is a nontrivial amplifier $P$ logical. Every $\omega\in\ker G_P$ with $\omega^\intercal a=1$ then gives a nontrivial base logical $M\omega$.
\end{lemma}

\begin{proof}
    Choose $B$ with $BG_{\bar P}^\intercal=U$ and add the first stabilizer in Eq.~\eqref{eq:common-central-stabilizers}, obtaining $(M',V,0)$ with $M'G_{\bar P}^\intercal=0$. Let the rows of $L$ supplement $\operatorname{row}G_P$ to $\ker G_{\bar P}$ and write $M'=CG_P+NL$. Independence of the rows of $G_P$ together with $L$ and the first normalizer equation give $V=H_{\bar P}C$ and $H_{\bar P}N=0$. Adding the second stabilizer reduces the pattern to $(NL,0,0)$. Some column of $N$ is outside $\operatorname{im}H_P^\intercal$: otherwise $N=H_P^\intercal R$ would make this pattern a stabilizer with $B=RL$. Choose $x\in\ker H_P$ with $N^\intercal x\neq0$. Then
    \begin{equation}
        a=M^\intercal x=G_P^\intercal C^\intercal x+L^\intercal N^\intercal x
        \label{eq:common-central-pairing-proof}
    \end{equation}
    has nonzero amplifier logical class. Finally, $H_{\bar P}M\omega=0$ and $x^\intercal M\omega=\omega^\intercal a=1$ show that $M\omega$ is a nontrivial base logical.
\end{proof}

Now, take any response decomposition of a nontrivial product logical with count $K$, and select the $D_P^A$ representatives supplied by Definition~\ref{def:transferable-amplifier-certificate} for the $a$ in the lemma. All $M\omega_i$ are nontrivial base logicals. An extended base bundle projects to $e_{\mathsf G_Q}(B)\delta_b^\intercal$ and contributes the entire zero-flag bundle to contraction $i$ when $(\omega_i)_b=1$. Disjointness places it in at most one contraction, at no greater count; no flagged base bundle is split. A new coupling bundle projects to $\delta_qv^\intercal$, with $v=e_{\mathsf G_A}(B)$, and contributes the singleton $\delta_q$ when $\omega_i^\intercal v=1$. Equation~\eqref{eq:methods-amplifier-parity} bounds the combined count of these singletons by $|B|\leq|\widehat B|$. A central singleton enters at most one contraction, and a newest-side singleton enters none. The resulting base decompositions therefore have counts $C_i$ obeying
\begin{equation}
    K\geq\sum_{i=1}^{D_P^A}C_i\geq D_P^A D_P^Q.
    \label{eq:common-certificate-transfer}
\end{equation}
This proves certificate transfer because the input decomposition was arbitrary. The history-reduction lemma gives the circuit lower bound.

The same parity count applied directly to an amplifier logical shows that $D_P^A$ is also a standalone fault-response certificate. However, its transferable definition has the additional disjointness restriction $D_P^A d_{\bar P}^A\leq n_A$, since each detecting representative has weight at least $d_{\bar P}^A$. Thus full-value transfer $D_P^A=d_P^A$ is an extra property of the chosen amplifier circuit, not a consequence of its standalone distance.

For the upper bound, choose an amplifier $P$ logical $\gamma_P$ of weight $d_P^A$ and a minimum-fault zero-record base history $F$. Copy $F$ into the inherited slices in $\operatorname{supp}\gamma_P$, with all new amplifier-core and newest-tail locations fault free. The copied history has $d_P^A|F|$ faults. After the base cores in round $r$, its central data response is $e_r\gamma_P^\intercal$. The amplifier cores then sample opposite-type parities $(e_r)_q(G_{\bar P}\gamma_P)_\nu=0$; same-type amplifier cores acquire no propagating syndrome component, and the flags are unaffected by the incoming errors. If $b$ is the residual propagating syndrome component of an inherited same-type check, it is identical in every active copy, so its combined newest-side emission is
\begin{equation}
    \sum_a b(\gamma_P)_a(G_{\bar P})_{\nu a} = b(G_{\bar P}\gamma_P)_\nu=0.
    \label{eq:common-tail-cancellation}
\end{equation}
The side~1 window completes this cancellation before the new opposite-type tail in side~2 samples that coordinate. The other side register receives no propagating same-type amplifier-syndrome component. Consequently newest side data are clean at each round boundary, and the inherited flags and deferred readouts reproduce the complete base record. Induction over rounds gives zero complete record and final response
\begin{equation}
    \bigl(E_P(F) \gamma_P^{\intercal},0,0\bigr).
\label{eq:common-lifted-response}
\end{equation}
Opposite-type logical representatives detecting the two factors have a central tensor product pairing to one with this response, so it is nontrivial. This proves the upper bound in Theorem~\ref{prop:common-joint-amplification}, and equality follows when its lower and upper endpoints coincide.

Finally, the transferred certificate applies to every decomposition and hence is a valid base-code fault-response certificate at the next level. Retaining all older data, couplings, and flags and preserving the SJ interface proves Eq.~\eqref{eq:mixed-circuit-bounds} by induction. For the scalar statement, both sectors have lower bound $d_A^tD_0$, while an attained history in either seed sector lifts to a matching upper bound in that sector. This proves Eq.~\eqref{eq:common-scalar-recursion} without assuming that the seed certificate is attained separately in both sectors.

\subsection{Computing fault-response certificates}
\label{app:compute-fault-response-certificate}

Consider a CSS extraction circuit on $n$ data qubits, with check matrices $H_X$ and $H_Z$. For $P\in\{X,Z\}$, let $D_P^{\max}$ be the largest certificate value allowed by Definition~\ref{def:joint-response-certificate}. We compute it as the minimum assigned fault count of a response decomposition producing a nontrivial $P$ logical.

Each catalog entry $(e_i,c_i)$ consists of a data response $e_i\in\mathbb F_2^n$ and a positive integer cost $c_i$ equal to its assigned fault count. Include all weight-one data responses at cost one and responses of allowed pure-$P$ fault bundles in $P$-check gadgets at their fault counts. Each bundle must have zero net response on every flag, although contributions from individual faults may cancel. For exactness, the catalog must be complete: every admissible $r$-fault bundle response must be expressible as a sum of catalog responses with total cost at most $r$.

For an unflagged weight-$w$ $P$-check gadget, let $(a_1,\ldots,a_w)$ denote its data-qubit labels in CNOT order, and let $\mathbf b_a\in\mathbb F_2^n$ be the unit vector supported on qubit $a$. A single pure-$P$ fault immediately after the $j$th CNOT has a data response in
\begin{equation}
     \begin{aligned}
         \mathcal E_j
         &=\{\mathbf b_{a_j},h_j,\mathbf b_{a_j}+h_j\},\\
         h_j&=\sum_{\ell=j+1}^{w}\mathbf b_{a_\ell},
         \qquad j=1,\ldots,w.
     \end{aligned}
     \label{eq:dp-bare-responses}
\end{equation}
Each entry has cost one. Under the stated fault model, collecting these responses from all unflagged gadgets, together with the weight-one responses, gives a complete catalog: multi-fault data responses are sums of single-fault responses.

To reduce the search space, identify responses that differ by a $P$-type stabilizer in $S_P=\operatorname{im}H_P^{\mathsf T}$. Put $H_P$ into reduced row-echelon form over $\mathbb F_2$ and clear the pivot coordinates of each response $e$ by adding the corresponding reduced rows. This defines a linear map $\rho_P:\mathbb F_2^n\to\mathbb F_2^n$ selecting one representative per class, with
\begin{equation}
     \rho_P(e)=\rho_P(e')
     \quad\Longleftrightarrow\quad e+e'\in S_P.
     \label{eq:dp-response-equivalence}
\end{equation}
Set $q_i=\rho_P(e_i)$, discard zero representatives, and retain a minimum-cost entry for each distinct nonzero representative. Stabilizer reduction preserves the data syndrome, and these simplifications preserve the minimum cost of a nontrivial logical response. Let $m$ be the number of retained entries and relabel them as $\{(q_i,c_i)\}_{i=1}^{m}$, keeping their original realizations for witness reconstruction.

Construct a weighted graph with vertex set $\mathcal U_P=\rho_P(\mathbb F_2^n)$. Each vertex $u$ represents an accumulated data response modulo $S_P$; adding catalog entry $i$ gives an edge $u\to u+q_i$ of weight $c_i$. Thus, $\operatorname{dist}(0,u)$ is the minimum total catalog cost of reaching $u$. Every path yields an admissible response decomposition, and completeness ensures that every admissible decomposition has a catalog representation of no greater cost. Consequently,
\begin{equation}
 D_P^{\max}
 =\min_{\substack{u\in\mathcal U_P\setminus\{0\}\\
                  H_{\bar P}u=0}}
   \operatorname{dist}(0,u),
 \label{eq:dp-exact-distance}
\end{equation}
where $\bar P$ denotes the opposite CSS sector. Here $H_{\bar P}u$ is the algebraic syndrome of the accumulated data response; requiring it to vanish for $u\ne0$ selects a nontrivial logical class. Intermediate states and ordinary syndrome measurement outcomes remain unrestricted. The search computes a fault-response certificate, while circuit distance additionally requires a complete zero-record circuit history.

Algorithm~\ref{alg:dp-exact} finds the minimum using Dijkstra's algorithm; breadth-first search suffices when all $c_i=1$. The search space contains $2^{n-\operatorname{rank}H_P}$ response states, so the search has exponential worst-case complexity in the quotient dimension.

\begin{figure}[t]
\begin{minipage}{0.47\textwidth}
\raggedright
\small

\refstepcounter{frcalgorithm}
\label{alg:dp-exact}

\hrule
\smallskip

\noindent
\textbf{Algorithm \thefrcalgorithm. Exact computation of $D_P^{\max}$}

\smallskip

\begin{algorithmic}[1]
\Require Complete reduced catalog $\{(q_i,c_i)\}_{i=1}^{m}$ and $H_{\bar P}$
\Ensure $D_P^{\max}$

\State Initialize $d(0)\gets 0$
\State Initialize $d(u)\gets+\infty$ for every other vertex $u$
\State Initialize a minimum-priority queue $Q\gets\{0\}$, with priorities given by $d$

\While{$Q\neq\varnothing$}
    \State Extract from $Q$ the response state $u$ having minimum tentative cost $d(u)$

    \If{$u\neq0$ \textbf{and} $H_{\bar P}u=0$}
        \State \Return $d(u)$
    \EndIf

    \For{$i=1,\ldots,m$}
        \State $v\gets u+q_i$
        \State $C\gets d(u)+c_i$

        \If{$C<d(v)$}
            \State $d(v)\gets C$

            \If{$v\in Q$}
                \State Decrease the priority of $v$ to $C$
            \Else
                \State Insert $v$ into $Q$ with priority $C$
            \EndIf
        \EndIf
    \EndFor
\EndWhile

\State \Return $+\infty$ \Comment{No logical class is reachable}
\end{algorithmic}

\smallskip
\hrule
\end{minipage}
\end{figure}

Table~\ref{tab:dp-example-results} compares the certificates for two $\code{18,4,4}$ unflagged extraction circuits and two $\code{4,2,2}$ extraction circuits. For the $\code{18,4,4}$ examples, $\operatorname{rank}H_X=\operatorname{rank}H_Z=7$, giving $2048$ response states modulo stabilizers per sector. The two $\code{4,2,2}$ schedules have the same circuit distance but different fault-response certificates.

\begin{table}[t]
\centering
\small
\setlength{\tabcolsep}{3pt}
\caption{Exact fault-response certificates for the specified circuits. The first two rows refer to $\code{18,4,4}$. Circuit distances are reported for comparison and are not computed by the response search.}
\label{tab:dp-example-results}
\begin{tabular}{@{}lccc@{}}
\toprule
Circuit & $d_{\mathrm{circ}}$ & $D_X^{\max}$ & $D_Z^{\max}$ \\
\midrule
IBM schedule~\cite{bravyi2024high} & 3 & 2 & 3 \\
Alternative unflagged schedule & 4 & 3 & 3 \\
Unflagged parallel $\code{4,2,2}$ & 2 & 1 & 1 \\
Flagged $\code{4,2,2}$ & 2 & 2 & 2 \\
\bottomrule
\end{tabular}
\end{table}

\section{The flagged four-copy distance amplifier}
\label{app:Proof_for_flagged_422}

This appendix establishes code-distance doubling and circuit-distance amplification for the flagged $\code{4,2,2}$ construction. Complementary logical representatives provide two nontrivial base-code contractions, while the local flag bound allows the terms of an admissible product-code fault-response decomposition to be allocated between them without increasing the total assigned fault count. This transfers a base-code fault-response certificate from $D_P$ to $2D_P$. We then verify that the circuit assumptions are preserved under recursion and apply the general circuit-distance bounds to obtain exact recursive doubling under the seed condition in Eq.~\eqref{eq:iva-seed-condition}. Throughout, we use the circuit and fault model of Appendix~\ref{app:circuit_distance_amplification} and the SJ schedule of Method~\ref{subsec:syndrome-circ-amplified}.

\subsection{Logical contractions and the local flag bound}

Two complementary pair contractions detect each product logical, while a flag limits which correlated errors one amplifier fault can contribute to them. The amplifier is specified by
\begin{equation}
    A_1=\mathbb F_2^4,\quad A_0=A_2=\mathbb F_2,
    \qquad G_X=G_Z=(1,1,1,1).
    \label{eq:iva-amplifier}
\end{equation}
Label its four data coordinates by $e_0,e_1,e_2,e_3$ and its $X$ and $Z$ checks by $e_X,e_Z$, respectively. The product data coordinates are $q\otimes e_a$, $z\otimes e_X$, and $x\otimes e_Z$, where $q$, $z$, and $x$ range over the coordinate bases of $Q_1$, $Q_2$, and $Q_0$, respectively.

The new amplifier $X$ check $q\otimes e_X$ uses the flagged core in Fig.~\ref{fig:flagged422-twenty-faults}(a). The corresponding $Z$-check core is obtained by exchanging $X$ and $Z$, including the ancilla preparations and measurement bases, and reversing all CNOT directions. These cores are executed in the amplifier window of Fig.~\ref{fig:css_tensor_overview}(b). The two flag CNOTs bracket the data couplings at $a=1,2$. For either check type, the flag is measured after the coupling at $a=3$ and before the newly appended side couplings; the syndrome ancilla is measured only after its side~2 couplings are complete. The local gate orders are interleaved according to the prescribed SJ schedule.

We first identify the two logical contractions that will receive the allocated fault responses. The contraction argument in the proof of Lemma~\ref{lem:distance-amplification} specializes to the following complementary-pair statement.

\begin{corollary}[Complementary-pair contractions]
\label{lem:iva-logical-pairing}
Let $E$ be a nontrivial $P$ logical operator of $\TTT$, with $P\in\{X,Z\}$, and write its component in $Q_1\otimes A_1$ as $\sum_{a=0}^3 E_a\otimes e_a$, where $E_a\in Q_1$. There exists a pairing $ab\mid cd$, with $\{a,b,c,d\}=\{0,1,2,3\}$, such that $E_a+E_b$ and $E_c+E_d$ are base-code $P$ logical operators satisfying
\begin{equation}
    [E_a+E_b]=[E_c+E_d]\ne0.
    \label{eq:iva-complementary-pairing}
\end{equation}
Here, the brackets denote equivalence modulo base-code $P$-type stabilizers. The statement allows arbitrary support of $E$ on the side data qubits.
\end{corollary}

\begin{proof}
Choose dual logical bases of the amplifier as
\[
\begin{aligned}
    \ell_{X,1}&=e_0+e_1,&\ell_{X,2}&=e_0+e_2,\\
    \ell_{Z,1}&=e_0+e_2,&\ell_{Z,2}&=e_0+e_1.
\end{aligned}
\]
For each $P\in\{X,Z\}$ and $i=1,2$, define the representative family
\[
    \Omega_{P,i}:=\left\{\ell_{P,i},\,
    \ell_{P,i}+\sum_{a=0}^3e_a\right\}.
\]
This is our choice of the family $\Omega_{P,i}$ appearing in Lemma~\ref{lem:distance-amplification}. Its two members represent the same logical class $[\ell_{P,i}]$, since their difference $\sum_{a=0}^3e_a$ is the support of the amplifier $P$ stabilizer. Each member has weight two, and their supports are complementary pairs of amplifier coordinates.

By the contraction argument in the proof of Lemma~\ref{lem:distance-amplification}, every nontrivial product $P$ logical has an index $i$ for which contraction against either member of $\Omega_{\bar P,i}$ gives the same nonzero base-code logical class. If these two representatives have supports $\{a,b\}$ and $\{c,d\}$, the contractions defined in Eq.~\eqref{eq:app-logical-contraction} are precisely $E_a+E_b$ and $E_c+E_d$. This proves Eq.~\eqref{eq:iva-complementary-pairing}. The contraction argument applies to the original logical representative, including any support on the side data.
\end{proof}

The above representative families have $\Sigma_X=\Sigma_Z=2$ and $o_X=o_Z=1$. Since $d_X^A=d_Z^A=2$, Lemma~\ref{lem:distance-amplification} gives
\[
    d_P(\TTT)=2d_P(\QQQ),\qquad P\in\{X,Z\}.
\]
The K\"unneth identification in Eqs.~\eqref{eq:app-kunneth-homology-group} and~\eqref{eq:app-kunneth-cohomology-group} also gives $k(\TTT)=2k(\QQQ)$.

The complementary contractions identify nontrivial logical responses. To obtain a circuit-distance bound, we must also control the assigned fault counts of admissible decompositions of these responses. For a new amplifier check, the following local bound separates a zero-flag central response into a contribution common to all four base-code copies and a residual error whose weight does not exceed the fault count. The common contribution cancels in either pair sum, allowing the residual terms to be assigned to individual copies.

\begin{figure}[t]
\centering
\includegraphics[width=0.48\textwidth]{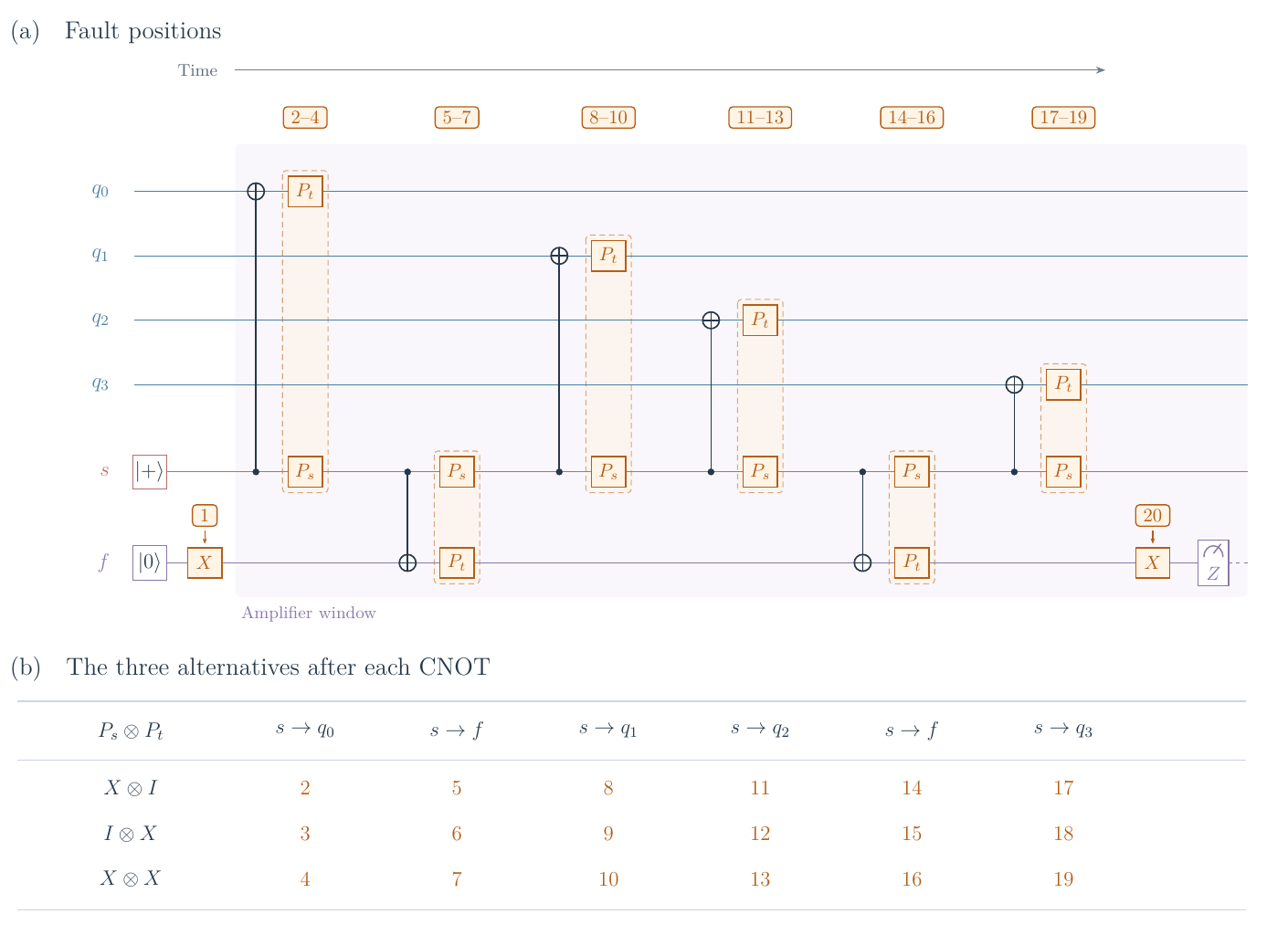}
\caption{\textbf{Flagged amplifier $X$-check core and its 20 single-fault cases.} (a) Core circuit and fault positions before the newly appended side data couplings, with $q_a=q\otimes e_a$, $s=q\otimes e_X$, and $f=f_X$. Cases 1 and 20 are $X$ errors immediately after flag preparation and before its $Z$ measurement, respectively. Each pair of orange boxes represents one CNOT fault, with the Pauli factors ordered as control and target. (b) Cases 2--19 correspond to $X\otimes I$, $I\otimes X$, and $X\otimes X$ after each of the six CNOTs. Each numbered case is considered separately, giving 20 cases at eight physical locations. The syndrome-preparation fault is a $Z$ error and is excluded from this pure-$X$ list.}
\label{fig:flagged422-twenty-faults}
\end{figure}

\begin{lemma}[Local flag bound]
\label{lem:iva-local-flag}
Consider the flagged amplifier $X$-check core in Fig.~\ref{fig:flagged422-twenty-faults}(a), or its $Z$ counterpart defined above. For $P\in\{X,Z\}$, restrict the $P$-check circuit to the operations from ancilla preparation through flag measurement. Suppose pure-$P$ faults occur at $r$ physical locations in this segment, and let $v\in\mathbb F_2^4$ denote their data response on the four qubits $q\otimes e_a$, ordered by $a=0,1,2,3$. If the flag response is zero, then there exist $v'\in\mathbb F_2^4$ and $b\in\mathbb F_2$ such that
\begin{equation}
v=v'+b(1,1,1,1),\qquad\operatorname{wt}(v')\le r.
\label{eq:iva-flag-bound}
\end{equation}
\end{lemma}

\begin{proof}
We prove the $X$ case by considering $r=0$, $r=1$, and $r\ge2$.

For $r=0$, the response vanishes, so we take $v'=0$ and $b=0$.

For $r=1$, the allowed pure-$X$ faults are the 20 cases listed in Fig.~\ref{fig:flagged422-twenty-faults}. Propagating these faults through the remaining CNOTs gives zero flag response precisely for cases
\[
    2,\ 3,\ 4,\ 7,\ 9,\ 12,\ 14,\ 17,\ 18,\ 19.
\]
Their data responses form the set
\[
    v\in\{0000,1111,1000,0100,0010,0001,0111\}.
\]
For every vector in this set, either $v$ or $v+(1,1,1,1)$ has weight at most one. Choosing the corresponding value of $b$ and setting $v'=v+b(1,1,1,1)$ proves Eq.~\eqref{eq:iva-flag-bound} for $r=1$.

For $r\ge2$, every $v\in\mathbb F_2^4$ satisfies
\[
    \begin{aligned}
    \min_{b\in\mathbb F_2}\operatorname{wt}\bigl(v+b(1,1,1,1)\bigr)
    &=\min\{\operatorname{wt}(v),4-\operatorname{wt}(v)\}\\
    &\le2\le r.
    \end{aligned}
\]
Thus, a suitable $b$ always exists. This argument applies to the combined response of the entire fault bundle, including cases in which individual flag responses cancel or faults affect the flag readout.

Exchanging $X$ and $Z$ and reversing all CNOT directions proves the $Z$ case.
\end{proof}

\subsection{Certificate transfer and recursive circuit bounds}

We now allocate the terms of an arbitrary admissible product-code decomposition. The construction below keeps inherited base-check bundles intact and distributes the residual amplifier responses among the four copies. The nontrivial complementary pairing is then selected using Corollary~\ref{lem:iva-logical-pairing}.

\begin{lemma}[Four-qubit fault allocation]
\label{lem:iva-event-doubling}
Fix $P\in\{X,Z\}$. Every admissible decomposition of a nontrivial $P$ logical operator of $\TTT$, with assigned fault count $K$, induces admissible decompositions of two nontrivial base-code $P$ logical operators with assigned fault counts $K_1,K_2$ satisfying
\[
    K_1+K_2\le K.
\]
Thus, the fault-allocation bound underlying Eq.~\eqref{eq:common-certificate-transfer} holds with two nontrivial base-code contractions.
\end{lemma}

\begin{proof}
Let $E$ be the product logical operator. Partition the bundle terms in its given admissible decomposition into inherited base-check bundles, indexed by $\Lambda_B$, and new amplifier-check bundles, indexed by $\Lambda_A$. Write
\begin{eqnarray}
    E&=&\sum_{\lambda\in\Lambda_B\cup\Lambda_A}
    e_{\widetilde{\mathsf G}_\lambda}(\widehat F_\lambda)
    +\sum_{\mu=1}^{s}u_\mu,\\
    K&=&\sum_{\lambda\in\Lambda_B\cup\Lambda_A}c_\lambda+s,
    \qquad c_\lambda:=|\widehat F_\lambda|.
\end{eqnarray}
Here, each $\widehat F_\lambda$ is a zero-flag pure-$P$ bundle, and each $u_\mu$ has weight one, as in Definition~\ref{def:joint-response-certificate}. All response sums below are over $\mathbb F_2$, whereas fault counts are summed over the nonnegative integers.

Let $\pi$ denote the projection onto the central register, and write
\begin{equation}
    \pi E=\sum_{a=0}^{3}E_a\otimes e_a,\qquad E_a\in Q_1.
\end{equation}
First, we construct a common vector $B\in Q_1$ and an admissible base-code decomposition of each $E_a+B$. The construction will apply simultaneously to all four copies, before choosing the two logical contractions.

For $\lambda\in\Lambda_B$, let $a_\lambda$ index the base-code copy containing the inherited check. Restrict $\widehat F_\lambda$ to the locations of the copied base gadget $\mathsf G_\lambda$, obtaining $F_\lambda$. This restriction retains all couplings belonging to $\mathsf G_\lambda$ and the corresponding syndrome readout; only faults in the newest side-coupling tail are removed. By the SJ ordering in Method~\ref{subsec:syndrome-circ-amplified}, these removed faults occur after the central data couplings and inherited flag measurements. Consequently,
\begin{equation}
    \begin{aligned}
    \pi e_{\widetilde{\mathsf G}_\lambda}(\widehat F_\lambda)
    &=e_{\mathsf G_\lambda}(F_\lambda)\otimes e_{a_\lambda},\\
    f_{\mathsf G_\lambda}(F_\lambda)
    &=f_{\widetilde{\mathsf G}_\lambda}(\widehat F_\lambda)=0,\\
    |F_\lambda|&\le c_\lambda.
\end{aligned}
\end{equation}
Thus, $e_{\mathsf G_\lambda}(F_\lambda)$ is an admissible base-code bundle term in copy $a_\lambda$, with assigned fault count $|F_\lambda|$. We retain $F_\lambda$ as a single bundle.

For $\lambda\in\Lambda_A$, let $q_\lambda\in Q_1$ be the physical-coordinate unit vector specifying the four central qubits $q_\lambda\otimes e_a$ coupled to this amplifier check. Restrict $\widehat F_\lambda$ to the gadget segment ending at the flag measurement, and denote its fault count by $r_\lambda\le c_\lambda$. The discarded tail faults affect neither the central response nor the flag outcome, so this restricted bundle still has zero flag response. Lemma~\ref{lem:iva-local-flag} therefore gives
\[
\begin{aligned}
v_\lambda&=b_\lambda\mathbf 1+v'_\lambda,
\qquad\mathbf 1=(1,1,1,1)^\intercal,\\
\operatorname{wt}(v'_\lambda)&=\sum_{a=0}^{3}v'_{\lambda,a}
\le r_\lambda\le c_\lambda,
\end{aligned}
\]
where $b_\lambda,v'_{\lambda,a}\in\{0,1\}$. In particular,
\[
\pi e_{\widetilde{\mathsf G}_\lambda}(\widehat F_\lambda)
=\sum_{a=0}^{3}(b_\lambda+v'_{\lambda,a})q_\lambda\otimes e_a.
\]
The contribution $b_\lambda q_\lambda$ is common to all four copies. Each nonzero $v'_{\lambda,a}$ supplies a weight-one base-code term $q_\lambda$ in copy $a$, and the total assigned fault count of these terms is at most $c_\lambda$.

For the weight-one terms already present in the original decomposition, assign a term supported on $q\otimes e_a$ to copy $a$ as the base-code term $q$, with assigned fault count one. Terms supported on the newest side data have zero central projection. Let
\[
u_{a,1},\ldots,u_{a,s_a}\in Q_1
\]
be the resulting weight-one terms in copy $a$. Then
\[
\sum_{a=0}^{3}s_a\le s.
\]

Define the accumulated common contribution by
\[
B:=\sum_{\lambda\in\Lambda_A}b_\lambda q_\lambda.
\]
Collecting the central responses yields, for every $a$,
\[
\begin{aligned}
E_a+B&=\sum_{\substack{\lambda\in\Lambda_B\\a_\lambda=a}}
e_{\mathsf G_\lambda}(F_\lambda)\\
&\quad+\sum_{\lambda\in\Lambda_A}v'_{\lambda,a}q_\lambda
+\sum_{\mu=1}^{s_a}u_{a,\mu}.
\end{aligned}
\]
Every term on the right is either a retained zero-flag base bundle response or a weight-one vector. This is therefore an admissible decomposition of $E_a+B$, with assigned fault count
\[
C_a:=\sum_{\substack{\lambda\in\Lambda_B\\a_\lambda=a}}|F_\lambda|
+\sum_{\lambda\in\Lambda_A}v'_{\lambda,a}+s_a.
\]
Summing over the four copies gives the budget
\begin{equation}
\begin{aligned}
\sum_{a=0}^{3}C_a
&=\sum_{\lambda\in\Lambda_B}|F_\lambda|
+\sum_{\lambda\in\Lambda_A}\operatorname{wt}(v'_\lambda)
+\sum_{a=0}^{3}s_a\\
&\le\sum_{\lambda\in\Lambda_B}c_\lambda
+\sum_{\lambda\in\Lambda_A}c_\lambda+s\\
&=K.
\end{aligned}
\label{eq:iva-column-budget}
\end{equation}

It remains to select the logical contractions. Applied to the full logical operator $E$, Corollary~\ref{lem:iva-logical-pairing} supplies a complementary pairing $ab\mid cd$ such that
\[
e_1:=E_a+E_b,\qquad e_2:=E_c+E_d,\qquad[e_1]=[e_2]\ne0.
\]
The allocation above is independent of this choice of pairing. Since $B+B=0$,
\[
\begin{aligned}
e_1&=(E_a+B)+(E_b+B),\\
e_2&=(E_c+B)+(E_d+B).
\end{aligned}
\]
Concatenating the constructed decompositions for the two copies in each pair, while keeping every base bundle intact, gives admissible decompositions of $e_1$ and $e_2$. Their assigned fault counts are
\[
K_1=C_a+C_b,\qquad K_2=C_c+C_d.
\]
Because the complementary pairs partition the four copies, Eq.~\eqref{eq:iva-column-budget} implies
\[
K_1+K_2=\sum_{a=0}^{3}C_a\le K.
\]
This establishes the required allocation for either choice of $P$.
\end{proof}

Lemma~\ref{lem:iva-event-doubling} supplies the fault-allocation property required by Theorem~\ref{prop:common-joint-amplification}. To obtain recursive bounds, we must additionally verify that the amplified circuit satisfies the same interface assumptions when it becomes the next base circuit.

\begin{theorem}[Flagged circuit-distance amplification]
\label{thm:iva-flagged-doubling}
Consider the flagged $\code{4,2,2}$ construction specified above, under the SJ schedule and the common circuit assumptions of Appendix~\ref{app:circuit_distance_amplification}. Fix a sector $P\in\{X,Z\}$. If the base-code circuit $\mathsf S$ has a fault-response certificate at $D_P$, then the amplified circuit $\widetilde{\mathsf S}$ has a certificate at $2D_P$, and
\begin{equation}
2D_P\le d_{\mathrm{circ},P}(\widetilde{\mathsf S})\le 2d_{\mathrm{circ},P}(\mathsf S).
\label{eq:iva-one-step-circuit-bounds}
\end{equation}
Starting from a seed certificate at $D_P^{(0)}$, the recursive circuit $\mathsf S^{(t)}$ has a certificate at $2^tD_P^{(0)}$ and satisfies
\begin{equation}
\begin{aligned}
2^tD_P^{(0)}&\le d_{\mathrm{circ},P}(\mathsf S^{(t)})\\
&\le 2^t d_{\mathrm{circ},P}(\mathsf S^{(0)})
\end{aligned}
\label{eq:iva-recursive-circuit-bounds}
\end{equation}
for every integer $t\ge0$. These bounds coincide whenever
\[
D_P^{(0)}=d_{\mathrm{circ},P}(\mathsf S^{(0)}).
\]
If both seed sectors admit certificates at the common value
\begin{equation}
D_X^{(0)}=D_Z^{(0)}=d_{\mathrm{circ}}(\mathsf S^{(0)}),
\label{eq:iva-seed-condition}
\end{equation}
then the overall circuit distance obeys the exact doubling relation in Eq.~\eqref{eq:flagged-circuit-distance-recursion}.
\end{theorem}

\begin{proof}
Lemma~\ref{lem:iva-event-doubling} establishes the required fault-allocation property for two nontrivial base-code contractions. Since the amplifier has $d_X^A=d_Z^A=2$, Theorem~\ref{prop:common-joint-amplification} gives both the transferred certificate at $2D_P$ and the one-step bounds in Eq.~\eqref{eq:iva-one-step-circuit-bounds}.

To apply the same result recursively, we verify that the amplified circuit remains within the common interface. The entire amplified code becomes the next base code: all three data registers, including the two side registers, and all product check rows are retained. Their extraction gadgets retain all existing data couplings, internal CNOT orders, and individual flag readouts. At the next tensor step, the newly appended side couplings occur after the inherited flag measurements, while the corresponding syndrome-auxiliary readouts are deferred until these couplings are complete. Each check continues to use a private syndrome ancilla, and the valid amplified extraction schedule serves as the next base schedule.

Every product data coordinate is incident on a measured check. Central data couple to the new amplifier checks. Since $G_X=G_Z=(1,1,1,1)$, each newest side datum in $Q_2\otimes A_0$ couples to four inherited $Z$ checks, and each newest side datum in $Q_0\otimes A_2$ couples to four inherited $X$ checks. Together with the preservation of the inherited gadgets and the prescribed ordering of the new couplings, this verifies that the common interface and the hypotheses of Lemma~\ref{lem:iva-event-doubling} remain valid at every recursive level.

Equation~\eqref{eq:mixed-circuit-bounds}, specialized to $D_P^{\AAA_j}=d_P^{\AAA_j}=2$ at every level, now gives Eq.~\eqref{eq:iva-recursive-circuit-bounds}, with the propagated certificate $2^tD_P^{(0)}$. Sectorwise equality follows when the corresponding seed certificate is attained. Under Eq.~\eqref{eq:iva-seed-condition}, the scalar recursion in Eq.~\eqref{eq:common-scalar-recursion} gives Eq.~\eqref{eq:flagged-circuit-distance-recursion}.
\end{proof}

\section{Canonical HGP amplifiers}
\label{app:hgp-amplification}

For HGP amplifiers, disjoint bit--bit slices certify the code-distance gain and also separate the circuit-fault contributions. Let $M_i\in\mathbb F_2^{r_i\times n_i}$ have full row rank, with $r_i<n_i$, $k_i=n_i-r_i$, and kernel distance $d_i$. We use the literal canonical presentation~\cite{ManesClaes2025}
\begin{equation}
    \begin{aligned}
        G_X&=(M_1\otimes I_{n_2}\mid I_{r_1}\otimes M_2^\intercal),\\
        G_Z&=(I_{n_1}\otimes M_2\mid M_1^\intercal\otimes I_{r_2}).
    \end{aligned}
    \label{eq:hgp-app-checks}
\end{equation}
Both check matrices have full row rank, and
\begin{equation}
    \begin{aligned}
        n_A&=n_1n_2+r_1r_2,&k_A&=k_1k_2,\\
        m_X^A&=r_1n_2,&m_Z^A&=n_1r_2.
    \end{aligned}
    \label{eq:hgp-amplifier-parameters}
\end{equation}
The two data sectors internal to the HGP amplifier are indexed by bit--bit pairs $(i,j)$ and check--check pairs $(\alpha,\beta)$. Their slice structure supplies both the exact code distances and the transferable amplifier certificates.

\subsection{Logical slices and code distances}

In the following, for each nontrivial amplifier $X$ or $Z$ logical operator, we construct disjoint opposite-type logical representatives that anti-commute with it. The resulting slice contractions give lower bounds on the product code distances, and explicit tensor-product logical operators attain these bounds.

\begin{lemma}[Disjoint detecting slices]
\label{lem:hgp-app-slices}
Every nontrivial amplifier $Z$ logical admits $d_1$ detecting $X$ logical representatives supported on distinct bit--bit rows. Dually, every nontrivial amplifier $X$ logical admits $d_2$ detecting $Z$ representatives supported on distinct bit--bit columns.
\end{lemma}
\begin{proof}
Write an amplifier $Z$ pattern as $(W,T)$, with $W\in\mathbb F_2^{n_1\times n_2}$ and $T\in\mathbb F_2^{r_1\times r_2}$. The normalizer equation is
\begin{equation}
    M_1W+TM_2=0.
    \label{eq:hgp-amplifier-normalizer}
\end{equation}
If $W\lambda=0$ for all $\lambda\in\ker M_2$, then $W=LM_2$ for some $L$. Full row rank of $M_2$ gives $T=M_1L$, making $(W,T)$ a $Z$ stabilizer. Hence, a nontrivial logical admits $\lambda\in\ker M_2$ with $y=W\lambda\neq0$. Since $M_1y=0$, at least $d_1$ entries $y_i$ equal one. For any $d_1$ such rows, choose
\begin{equation}
    \omega_i=(\delta_i\lambda^\intercal,0).
    \label{eq:hgp-detecting-representatives}
\end{equation}
These have disjoint supports, satisfy $G_Z\omega_i=0$, and pair to one with $(W,T)$. The transposed argument exchanges $M_1,M_2$ and $X,Z$, giving $d_2$ detecting columns.
\end{proof}

Applying this lemma after the common central-pairing Lemma~\ref{lem:common-central-pairing}, every nontrivial product $Z$ logical $E$ has $d_1$ nontrivial base-code contractions
\begin{equation}
z_i=\sum_{j=1}^{n_2}\lambda_j E_{\bullet,i,j}
\label{eq:hgp-app-contraction}
\end{equation}
on distinct bit--bit rows. Here $E_{\bullet,i,j}$ is its central response at amplifier coordinate $(i,j)$, and the detecting collection, including $\lambda$, is chosen for the amplifier logical supplied by the common lemma. The dual statement uses $d_2$ columns. Both newest side registers may be nonzero; the common lemma already accounts for them.

\begin{proposition}[HGP code-distance amplification]
\label{prop:hgp-code-distance}
For every CSS base code with $k>0$,
\begin{equation}
\widetilde d_Z=d_1d_Z,\qquad\widetilde d_X=d_2d_X.
\label{eq:hgp-app-static-proof}
\end{equation}
\end{proposition}
\begin{proof}
The disjoint contractions give
\begin{equation}
    \operatorname{wt}(E)
    \geq\sum_i\sum_j\operatorname{wt}(E_{\bullet,i,j})
    \geq\sum_i\operatorname{wt}(z_i)\geq d_1d_Z,
\end{equation}
where $i$ ranges over the selected rows. For the upper bound, choose a weight-$d_1$ word $c\in\ker M_1$ and a unit vector $\delta_j$ with nonzero class in $\operatorname{coker}M_2^\intercal$. The amplifier pattern $(c\delta_j^\intercal,0)$ is a nontrivial $Z$ logical of weight $d_1$; its central tensor product with a minimum-weight base $Z$ logical attains $d_1d_Z$. The dual construction proves the $X$ identity.
\end{proof}

Taking the algebraic base code to be one unchecked logical qubit gives $d_Z^A=d_1$ and $d_X^A=d_2$; this makes no circuit-coverage assumption. The code-distance identity also follows by applying the two-term tensor-distance formula of Ref.~\cite[Theorem~17]{ZengPryadko2020} to the two classical factors of the amplifier.

\subsection{Transferable amplifier certificates and recursive closure}
\label{app:hgp-circuit-assumptions}

The literal HGP check presentation keeps each gadget response within one detecting slice, making the certificate transferable. Measure each literal row of Eq.~\eqref{eq:hgp-app-checks} with an unflagged private-ancilla gadget visiting each supported datum once, in any internal order compatible with an ideal-correct CSS schedule. Assume every amplifier datum has a physical check incidence, and use the SJ interface of Appendix~\ref{app:circuit_distance_amplification}, preserving the base gadgets and their permitted flags.

\begin{proposition}[Transferable amplifier certificates for HGP codes]
\label{prop:hgp-transferable-certificate}
These amplifier circuits satisfy Definition~\ref{def:transferable-amplifier-certificate} with
\begin{equation}
    D_Z^A=d_1=d_Z^A,\qquad D_X^A=d_2=d_X^A.
    \label{eq:hgp-transferable-certificates}
\end{equation}
\end{proposition}
\begin{proof}
Use the detecting representatives of Lemma~\ref{lem:hgp-app-slices}. A literal $Z$ check indexed by $(i,\beta)$ has bit--bit support only in row $i$. Every fault bundle within its gadget therefore affects at most one selected detecting parity; responses on the HGP amplifier's internal check--check sector meet none of the representatives. For a nonempty bundle, the parity weight is at most $1\leq|B|$, while an empty bundle has zero response. This proves Eq.~\eqref{eq:methods-amplifier-parity} for every bundle, independently of its internal fault correlations or gate order. The literal $X$ checks similarly meet only one bit--bit column, proving the dual condition.
\end{proof}

The SJ projection conditions hold because each constituent gadget is preserved and deleting its newest tail leaves its central response and earlier flag outcomes unchanged. No new HGP flags are required. The prescribed side-window order and deferred syndrome readouts give the common history lift. Theorem~\ref{prop:common-joint-amplification} therefore transfers base certificates $D_Z^Q,D_X^Q$ to $d_1D_Z^Q,d_2D_X^Q$ and gives the corresponding circuit-distance bounds.

At each recursive step, all older data, couplings, and flags remain in the base circuit, and the same canonical HGP presentation supplies the new amplifier gadgets. Coverage is preserved: central data meet the amplifier cores, and each newest side datum has a coupling associated with a nonzero amplifier check row. The same SJ timing and fault-location embeddings are retained, so Eq.~\eqref{eq:mixed-circuit-bounds} applies at every finite level, with equality when the seed certificate is attained in the sector under consideration. The symmetric scalar case follows from Eq.~\eqref{eq:common-scalar-recursion}. No HGP structure or check-row independence is required of the base code or its descendants. The full-row-rank classical seed matrices and literal measured generators are essential hypotheses here. Changing the amplifier generators requires a new verification of their response supports.

\section{Rotated surface amplifiers}
\label{app:rotated-surface-code-amplifiers}

We establish the transferable amplifier certificates of Definition~\ref{def:transferable-amplifier-certificate} for rotated surface codes. Disjoint horizontal and vertical logical representatives provide the detecting contractions, while the CNOT order bounds the number of their parities changed by each fault. This verifies the amplifier-specific input to Theorem~\ref{prop:common-joint-amplification}. We then apply Theorem~\ref{prop:common-joint-amplification} to obtain the amplified certificate and circuit-distance bounds at each amplification step. Throughout this appendix, we use the assumptions of Methods~\ref{subsec:circuit-distance-methods}. Let us write $d_{\mathrm{circ},P}(\mathsf S_Q)$ for the base circuit distance and $D_P^Q$ for its fault-response certificate.

\subsection{Logical rows and columns}

Horizontal and vertical logical strings furnish disjoint contractions for the two Pauli sectors. Let $\AAA_{h,w}$ be the standard rotated surface-code rectangle with $h$ rows and $w$ columns of data qubits, where $h,w\geq2$. Choose the boundary types so that logical $X$ strings run vertically and logical $Z$ strings horizontally. We use the literal checkerboard plaquette generators, with weight-two $X$ checks on the top and bottom boundaries and weight-two $Z$ checks on the left and right boundaries, continuing the plaquette coloring. The $hw-1$ generators are independent, so both $G_X$ and $G_Z$ have full row rank and $k_A=1$.

Write $z_r$ for the support vector of the horizontal logical $Z$ in row $r$, and $x_c$ for that of the vertical logical $X$ in column $c$. Then
\begin{equation}
    G_Xz_r=0,\qquad G_Zx_c=0,\qquad z_r^\intercal x_c=1.
    \label{eq:rs-detecting-rows-columns}
\end{equation}
The $h$ row supports are pairwise disjoint, as are the $w$ column supports. All $z_r$ represent the same logical $Z$ class, and all $x_c$ represent the same logical $X$ class. Since $k_A=1$, every nontrivial amplifier $X$ logical pairs to one with every $z_r$, and every nontrivial amplifier $Z$ logical pairs to one with every $x_c$. These collections therefore satisfy the disjointness and detection requirements of Definition~\ref{def:transferable-amplifier-certificate}.

\begin{lemma}[Row and column contractions]
\label{lem:rs-stripes}
    For every nontrivial product $X$ logical, its $h$ row contractions represent the same nontrivial base-code $X$ logical class. Dually, the $w$ column contractions of a nontrivial product $Z$ logical represent the same nontrivial base-code $Z$ logical class.
\end{lemma}

\begin{proof}
    Let $E_{r,c}\in\mathbb F_2^n$ be the central data pattern in the base-code copy indexed by amplifier datum $(r,c)$, and collect these vectors as the columns of $M$. The contractions are the modulo-two sums
    \begin{equation}
        Mz_r=\sum_{c=1}^{w}E_{r,c},\qquad
        Mx_c=\sum_{r=1}^{h}E_{r,c}.
        \label{eq:rs-contractions}
    \end{equation}
    For a nontrivial product $X$ logical, the central-pairing Lemma~\ref{lem:common-central-pairing} supplies a nontrivial amplifier $X$ logical detected by all $z_r$. Hence, every $Mz_r$ is a nontrivial base-code $X$ logical. This argument includes product logicals with nonzero support on either newest side register. 
    
    To compare their classes, write the product pattern as $(M,V,U)$ using Appendix~\ref{app:circuit_distance_amplification}. Its normalizer equations include $MG_Z^\intercal=H_X^\intercal U$. Since any two rows differ by amplifier $Z$ stabilizers, $z_r+z_s=G_Z^\intercal y$ for some $y$, giving
    \begin{equation}
        Mz_r+Mz_s=H_X^\intercal Uy.
        \label{eq:rs-stripe-equivalence}
    \end{equation}
    Thus, the contractions differ only by a base-code $X$ stabilizer. Exchanging $X,Z$ and rows, columns proves the dual statement.
\end{proof}

Disjointness also gives the code distances. The central weight in each row is at least the weight of its contraction, which is at least $d_X$ for a nontrivial product $X$ logical. Summing over rows gives $\widetilde d_X\geq h d_X$. A minimum-weight base $X$ logical tensored with a vertical amplifier logical attains this bound with zero newest-side support. The column argument is dual, so
\begin{equation}
    \widetilde d_X=h d_X,\qquad \widetilde d_Z=w d_Z.
    \label{eq:rs-code-distance}
\end{equation}
The same disjoint-representative argument on the amplifier itself gives $(d_X^A,d_Z^A)=(h,w)$. In particular, $\AAA_{d_A,d_A}$ has parameters $\code{d_A^2,1,d_A}$.

\subsection{Measurement order and transferable amplifier certificates}
\label{app:rs-assumptions}

We choose a CNOT order so that a single fault changes at most one detecting row or column parity~\cite{PhysRevA.90.062320,PhysRevApplied.8.034021}. Measure each literal amplifier check with a syndrome ancilla, visiting each supported datum once. An $X$ ancilla controls its data CNOTs, while a $Z$ ancilla receives them, as in Methods~\ref{subsec:syndrome-circ-amplified}. For a bulk plaquette, label the corners northwest, northeast, southwest, and southeast. Use the orders~\cite{PhysRevA.90.062320,PhysRevApplied.8.034021,GoogleThreshold}
\begin{align}
    X:&\quad\mathrm{NW},\mathrm{NE},\mathrm{SW},\mathrm{SE},
    \label{eq:rs-x-order}\\
    Z:&\quad\mathrm{NW},\mathrm{SW},\mathrm{NE},\mathrm{SE}.
    \label{eq:rs-z-order}
\end{align}
Thus, $X$ checks visit one row before the other, and $Z$ checks visit one column before the other. The two-qubit hooks are transverse to the corresponding logical strings. The proof below uses the local CNOT orders specified above, which ensure that a single fault changes at most one detecting row or column parity.

\begin{proposition}[Transferable amplifier certificates for rotated surface codes] \label{prop:rs-transferable-certificates}
These amplifier circuits satisfy Definition~\ref{def:transferable-amplifier-certificate} with
\begin{equation}
    D_X^A=h=d_X^A,\qquad D_Z^A=w=d_Z^A.
    \label{eq:rs-transferable-certificates}
\end{equation}
\end{proposition}

\begin{proof}
    For $P=X$, use the fixed detecting collection $z_1,\ldots,z_h$. It remains to verify the bundle-parity inequality. A pure-$X$ fault in an amplifier $X$-check gadget produces either a data singleton or a suffix of the ordered support. At a CNOT, an ancilla-only output fault gives the suffix after the current datum; a correlated ancilla--data $XX$ fault includes that datum in the suffix. Each is one faulty physical location. For the order in Eq.~\eqref{eq:rs-x-order}, the suffix parities, listed as (top row, bottom row), are
    \begin{equation}
        \begin{array}{c|ccccc}
        \text{suffix length}&0&1&2&3&4\\ \hline
        \text{row parities}
         &(0,0)&(0,1)&(0,0)&(1,0)&(0,0).
        \end{array}
        \label{eq:rs-parity}
    \end{equation}
    Thus, a single fault changes at most one detecting row parity. A singleton has the same property. Boundary $X$ checks lie within one row; preparation faults give a full-check response or zero, and readout faults give no data response. These cases obey the same bound.
    
    Binary linearity and sub-additivity of Hamming weight yield
    \begin{equation}
        \operatorname{wt}\!\left[
         \bigl(z_r^\intercal e_{\mathsf G_A}(B)\bigr)_{r=1}^{h}\right]
        \leq\sum_{f\in B}1=|B|.
        \label{eq:rs-bundle-parity}
    \end{equation}
    This is exactly the required parity inequality, for the same detecting collection and every amplifier $X$ gadget and bundle. 
    
    For $P=Z$, use $x_1,\ldots,x_w$. The column-first order in Eq.~\eqref{eq:rs-z-order} gives the same table for the left and right column parities, and boundary $Z$ checks lie within one column. The identical argument proves the bound for every pure-$Z$ bundle.
\end{proof}

The cancellation in Eq.~\eqref{eq:rs-parity} is the reason for the gate order. A horizontal two-qubit $X$ hook cancels in a row contraction, while a three-qubit suffix leaves only one odd row. The physical response may occupy two rows, but its contracted response requires at most one singleton term. The column statement for $Z$ faults is the same. No new flags are needed.

\subsection{SJ amplification and recursive closure}

We now apply the general circuit-distance amplification theorem using the surface-code bound proved above. We use the SJ construction of Methods~\ref{subsec:syndrome-circ-amplified}, under the conditions of Appendix~\ref{app:circuit_distance_amplification}, retaining the base-code circuits and their flags and using the amplifier gate order specified above. To repeat the amplification, we take the entire amplified code and its measurement circuit as the new base.

With this schedule, faults in the newly added side couplings do not affect the data in the base-code copies ($Q_1\otimes A_1$ in Eq.~\ref{eq:tensor-qubit-sectors} of Methods) or earlier flag outcomes. Each extended-check fault group contributes a base-code response to at most one row or column contraction. Proposition~\ref{prop:rs-transferable-certificates} shows that coupling-check faults change no more row or column parities than the number of faults. Applying the base certificate to these contractions then gives the amplified certificate. For the upper bound, we repeat an undetected base-code fault pattern in the copies selected by an amplifier logical operator. The SJ schedule ensures that this pattern remains undetected in the amplified circuit.

\begin{theorem}[Rotated-surface circuit amplification] \label{thm:rs-amplification}
    Let $\mathsf S_Q$ be a CSS base-code extraction circuit with fault-response certificates $D_X^Q,D_Z^Q$. Under the amplifier schedule and SJ conditions specified above, the amplified circuit $\widetilde{\mathsf S}$ inherits the certificates $hD_X^Q,wD_Z^Q$ and satisfies
    \begin{equation}
    \begin{aligned}
        hD_X^Q&\leq d_{\mathrm{circ},X}(\widetilde{\mathsf S})
         \leq h\,d_{\mathrm{circ},X}(\mathsf S_Q),\\
        wD_Z^Q&\leq d_{\mathrm{circ},Z}(\widetilde{\mathsf S})
         \leq w\,d_{\mathrm{circ},Z}(\mathsf S_Q).
    \end{aligned}
    \label{eq:rs-one-step}
    \end{equation}
    In a sector with $D_P^Q=d_{\mathrm{circ},P}(\mathsf S_Q)$, the bounds coincide and the transferred certificate is attained.
\end{theorem}

\begin{proof}
    The standard rotated-surface check matrices have full row rank, and Proposition~\ref{prop:rs-transferable-certificates} gives $D_P^A=d_P^A$. Together with the SJ conditions stated above, this satisfies the assumptions of Theorem~\ref{prop:common-joint-amplification}, which gives the amplified certificates and the stated circuit-distance bounds. For the upper bound, we repeat a minimum-fault undetected logical history of the base circuit in the $h$ copies selected by a vertical amplifier $X$ logical, or the $w$ copies selected by a horizontal amplifier $Z$ logical. The resulting history remains undetected and has $h$ or $w$ times as many faults, respectively. If the base-code certificate equals its circuit distance in a given sector, the lower and upper bounds coincide in that sector.
\end{proof}

For recursion, retain the entire amplified code and circuit as the next base, including all older side data, couplings, and flags. Only the newest side data are discarded in the next projection. The same local orders, flag timing, synchronized cores, side-window order, and fault-location embeddings are retained. Every central datum meets an amplifier check, and each newest side datum has an inherited-check coupling supplied by a nonzero amplifier check row, so physical check coverage is preserved. Neither surface-code structure nor check-row independence is required of the base code or its descendants.

Therefore, successive rectangular amplifiers $\AAA_{h_j,w_j}$ transfer the seed certificates to
\begin{equation}
    D_X^{(t)}=\left(\prod_{j=1}^{t}h_j\right)D_X^{(0)}, \qquad  D_Z^{(t)}=\left(\prod_{j=1}^{t}w_j\right)D_Z^{(0)}.
    \label{eq:rs-recursive-certificates}
\end{equation}
The corresponding circuit-distance bounds are the mixed-amplifier bounds of Methods~\ref{subsec:circuit-distance-methods}, with the same sector-wise gains. Let $t$ count amplification steps. For a fixed square amplifier and a seed certificate attained in sector $P$,
\begin{equation}
    d_{\mathrm{circ},P}(\mathsf S^{(t)}) =d_A^t d_{\mathrm{circ},P}(\mathsf S^{(0)}).
\label{eq:rs-exact-recursion}
\end{equation}
For a fixed square amplifier, the overall circuit distance multiplies exactly if both seed sectors admit a certificate at the common value $D_0=d_{\mathrm{circ}}(\mathcal S^{(0)})$, as discussed in Methods~\ref{subsec:circuit-distance-methods}. The seed's $X$- and $Z$-sector circuit distances need not be equal. Without this condition, the transferred certificates still give lower bounds at every amplification step, but do not establish exact circuit-distance multiplication.

\section{Logical coordinates, symmetry lifts, and surgery guarantees}
\label{app:logical-addressability}

This appendix establishes how the tensor construction carries logical coordinates, compatible physical permutations, and prescribed surgery constraints from a base code to its amplification. Our aim is to identify the inherited logical structure and determine which distance guarantees apply to the resulting merged codes. We use the conventions of Methods~\ref{subsec:addressability-methods}. Throughout, $H_0(\AAA)=H_2(\AAA)=0$, and maps between tensor complexes are formed before central CSS truncation.

\subsection{Inherited logical coordinates and symmetries}

We first verify that the tensor representatives retain the logical labels and dual pairings of the constituent codes. Their physical supports specify the inherited addresses used in the symmetry and surgery constructions below.

\begin{proposition}[Inherited logical coordinates] \label{prop:logical-basis-inheritance}
    For full dual logical bases of $\QQQ$ and $\AAA$, Eq.~\eqref{eq:canonical-basis-product} gives full dual logical bases of $\TTT$. For a base canonical set containing only $k_{\mathrm{can}}$ dual logical pairs, the same construction gives $k_{\mathrm{can}}k_A$ independent dual pairs. Their physical supports and recursive labels are given by Eq.~\eqref{eq:canonical-basis-product} and its recursive application.
\end{proposition}

\begin{proof}
    A tensor of two $Z$ cycles is a degree-two cycle, and a tensor of two $X$ cocycles is a degree-two cocycle. Their side components can be set to zero because each factor separately has zero syndrome. The K\"unneth isomorphism gives
    \begin{equation}
    H_2(\mathcal T)\cong H_1(\QQQ)\otimes H_1(\AAA),
    \label{eq:addressability-kunneth}
    \end{equation}
    and the dual statement holds for cohomology. Central truncation retains the maps into and out of $T_2$, so it preserves these logical spaces. The tensor representatives therefore span all $kk_A$ logical dimensions. Their pairing factors as
    \begin{equation}
        \begin{aligned}
        \widetilde\ell_{X,\mu a}^{\intercal}\widetilde\ell_{Z,\nu b}
        &=(\ell_{X,\mu}^{\intercal}\ell_{Z,\nu})
        ((\ell^A_{X,a})^{\intercal}\ell^A_{Z,b})\\
        &=\delta_{\mu\nu}\delta_{ab}.
        \end{aligned}
        \label{eq:product-logical-pairing}
    \end{equation}
    Restricting this identity pairing to a canonical subset proves independence of its $k_{\mathrm{can}}k_A$ pairs without assuming completeness.
    
    A binary tensor coordinate is nonzero exactly when both factor coordinates are nonzero, proving the Cartesian support formula. Reapplying the construction to the complete current code proves the recursive formula. This argument uses the prescribed tensor-and-truncate order, rather than an assumed associativity of truncated products. The representatives vanish on all newly introduced side data, although the checks and code space involve those data.
\end{proof}

For arbitrary representatives and Pauli types, the set identity $(S\times U)\cap(T\times V)=(S\cap T)\times(U\cap V)$ gives
\begin{equation}
    \begin{aligned}
    &\bigl|\widetilde S_{P,\mu a}\cap\widetilde S_{P',\nu b}\bigr|\\
    &\qquad=|S_{P,\mu}\cap S_{P',\nu}|\,
    |S^A_{P,a}\cap S^A_{P',b}|.
    \end{aligned}
    \label{eq:canonical-intersection-inheritance}
\end{equation}
Consequently, disjoint base supports remain disjoint after replication. If both factor bases have exactly one physical intersection for each conjugate pair and none for nonconjugate pairs, the product has the same property. Duality by itself only fixes intersection parity and does not exclude additional even intersections. In particular, different amplifier labels can have overlapping supports within the same inherited base region.

Canonical HGP and \textit{lifted product} (LP) bases supply structured representatives~\cite{XuZhouZheng2024Homological,ZhengZhengJiangXu2026CanonicalLP}; clustered bases provide another choice~\cite{QGPU}. If the amplifier also has row--column labels, the HGP-style product label can be written $(i,j,i_A,j_A)$. These are statements about the chosen coordinates and supports, with no claim of minimum weight or closure within an ordinary two-factor code family. In particular, the canonical LP definition permits canonical rows to span only a logical subspace~\cite[Definition~III.2]{ZhengZhengJiangXu2026CanonicalLP}; the remaining classes are still needed for a full logical basis.

Having identified the inherited coordinates, we next determine when a base physical permutation lifts to the amplified data. Compatibility with the retained check rows is essential because these rows also label side data.

\begin{proposition}[Lifting a compatible permutation] \label{prop:permutation-lift}
    Under Eq.~\eqref{eq:permutation-chain-condition}, the physical permutation in Eq.~\eqref{eq:lifted-data-permutation} preserves the product stabilizer code. Its $Z$-logical action is $\Pi_{1,*}\otimes I_{H_1(\AAA)}$; its $X$-logical action in dual coordinate bases is the inverse transpose.
\end{proposition}

\begin{proof}
    The compatibility equations ensure that $\Pi$ is a chain automorphism of $\QQQ$. Tensoring with the identity on $\AAA$ therefore gives a chain automorphism of the total complex, whose degree-two component is Eq.~\eqref{eq:lifted-data-permutation}. Its components in degrees one, two, and three are coordinate permutations, so it permutes the physical data and both product-check families, preserving the product stabilizer code.
    
    Under the K\"unneth identification, the induced homology map sends $[z]\otimes[\gamma]$ to $[\Pi_1z]\otimes[\gamma]$. Thus the $Z$-logical action is $\Pi_{1,*}\otimes I_{H_1(\AAA)}$. Since physical permutations preserve the binary $X$--$Z$ pairing, the action on the dual $X$-logical coordinates is the inverse transpose.
\end{proof}

\begin{corollary}[Inherited cyclic logical fibers]\label{cor:canonical-orbit-inheritance}
    Suppose compatible base permutations $\sigma_0,\sigma_1,\sigma_2$ act on chosen representatives as $\sigma_1\ell_{P,i,j,m}=\ell_{P,i,j,m+1}$, with $m$ taken modulo the fiber length. Then
    \begin{equation}
        \widetilde\sigma_2\widetilde\ell_{P,i,j,m,a}
        =\widetilde\ell_{P,i,j,m+1,a}.
        \label{eq:canonical-orbit-inheritance}
    \end{equation}
    If the base relation holds only modulo stabilizers, the product relation holds on logical classes.
\end{corollary}

\begin{proof}
    The central block $\sigma_1\otimes I$ applies the base cyclic shift while leaving the amplifier representative unchanged, and the side blocks preserve the zero side components. This proves the claimed identity when the base relation holds exactly.
    
    If the base relation holds modulo stabilizers, tensoring a base $Z$ boundary with an amplifier $Z$ cycle gives a total boundary. Dually, tensoring a base $X$ coboundary with an amplifier $X$ cocycle gives a total coboundary. The same cyclic shift therefore holds on the product logical classes.
\end{proof}

The lifted permutation acts in the same way at every amplifier label; independent actions at selected labels do not follow. Moreover, a base automorphism preserving only check row spaces may require nonpermutation changes of check basis, which need not give a physical permutation of the side data. A fault-tolerant circuit for the permutation, or for a fold-transversal gate, requires a separate implementation argument.

\subsection{Surgery lifts and distance guarantees}

We now use the inherited logical coordinates to specify which base logical constraints are promoted to stabilizers by surgery. We first consider the complete lift and establish the conditions under which it preserves the amplified memory's code distance. We then introduce the selective lift, compare the resulting logical spaces, and clarify the scope of the circuit guarantees.

For the inclusion $j:\QQQ\to\mathcal M=\operatorname{Cone}(f)$, the promoted space satisfies $\ker j_*=\mathcal W$. Indeed, $(z,0)=\partial^M(b,u)$ is equivalent to $\partial^Uu=0$ and $z=\partial^Qb+f(u)$. The cone exact sequence then gives
\begin{equation}
    0\longrightarrow H_1(\QQQ)/\mathcal W
    \longrightarrow H_1(\mathcal M)
    \longrightarrow\ker f_{*,0}\longrightarrow0.
    \label{eq:surgery-remaining-space}
\end{equation}
Thus, $\dim H_1(\mathcal M)=k-m+\kappa$. The quotient describes surviving data-logical classes; the additional $\kappa$ dimensions are not automatically protected data logicals.

The complete lift applies the base surgery to all amplifier logical directions by tensoring the entire chain map.

\begin{proposition}[Tensoring the complete surgery map] \label{prop:surgery-tensor-lift}
    For $F=f\otimes I_{\AAA}$, Eq.~\eqref{eq:cone-tensor-identity} is a coordinate-permutation chain isomorphism. The degree-two cone code promotes precisely $\mathcal W\otimes H_1(\AAA)$, of dimension $mk_A$, and has $(k-m+\kappa)k_A$ logical qubits.
\end{proposition}

\begin{proof}
    Both full complexes have degree-$r$ space
    \begin{equation}
        \bigoplus_{i+j=r}Q_i\otimes A_j\ \oplus\!\bigoplus_{i+j=r-1}U_i\otimes A_j.
    \end{equation}
    After regrouping summands, their differentials coincide: they contain the same factor differentials and the same off-diagonal map $f\otimes I$. Signs vanish over $\mathbb F_2$. Under K\"unneth, the induced degree-two map is $f_*\otimes I$ on $H_1(\mathcal U)\otimes H_1(\AAA)$, with image $\mathcal W\otimes H_1(\AAA)$. The cone exact sequence identifies this image with the promoted data-logical subspace. Applying K\"unneth to $\mathcal M\otimes\AAA$ gives the merged logical dimension $(k-m+\kappa)k_A$. Central truncation preserves degree-two homology.
\end{proof}

Since the complete lift retains the full tensor structure, the code-distance results apply with the merged code $\mathcal M$ as the base code. The following corollary gives the resulting distance bound and the condition under which the complete lift preserves the amplified memory's protection.

\begin{corollary}[Code distance of the complete lift] \label{cor:surgery-code-distance}
    Suppose $\mathcal M=\operatorname{Cone}(f)$ has at least one remaining logical qubit. Let $\mathcal C_{\mathcal M}=\operatorname{tr}_{2}\operatorname{Tot}(\mathcal M\otimes\AAA)$ denote the central CSS code of the complete lift. Under the hypotheses of Lemma~\ref{lem:distance-amplification},
    \begin{equation}
        d_P(\mathcal C_{\mathcal M})\geq\left\lceil\alpha_Pd_P(\mathcal M)\right\rceil,
        \qquad P\in\{X,Z\}.
        \label{eq:surgery-static-bound}
    \end{equation}
    For each of the three amplifier families with exact code-distance multiplication established in this work,
    \begin{equation}
        d_P(\mathcal C_{\mathcal M})=d_P^A d_P(\mathcal M).
        \label{eq:surgery-static-multiplication}
    \end{equation}
    If $d_P(\mathcal M)\geq d_P(\QQQ)$, the complete lift therefore retains the original amplified memory's certified lower bound. For an exact-distance family, it also satisfies $d_P(\mathcal C_{\mathcal M}) \geq d_P(\TTT)$.
\end{corollary}

\begin{proof}
    The cone isomorphism identifies the complete lift with the central tensor construction applied to $\mathcal M$. Apply Lemma~\ref{lem:distance-amplification}, or the applicable exact-distance result, with $\mathcal M$ as the base code. The distance here is that of the full stabilizer code, including any additional logical degrees of freedom counted by $\kappa$. It is not a dressed subsystem distance obtained by discarding them.
\end{proof}

To restrict the promoted subspace to one amplifier logical direction, we instead use the selective lift.

\begin{proposition}[Selecting one amplifier logical direction] \label{prop:selective-surgery-lift}
    For a nonzero amplifier class $[\gamma]$, the map in Eq.~\eqref{eq:selective-surgery-map} is a chain map. Its degree-two cone code promotes precisely $\mathcal W\otimes\operatorname{span}\{[\gamma]\}$, of dimension $m$, and has $kk_A-m+\kappa$ logical qubits.
\end{proposition}

\begin{proof}
    Because $G_X\gamma=0$ and $f$ is a chain map,
    \begin{equation}
        \partial F^\gamma(u)=\partial^Qf(u)\otimes\gamma = f(\partial^Uu)\otimes\gamma=F^\gamma(\partial u).
    \end{equation}
    The induced degree-two map sends $[u]$ to $f_*[u]\otimes[\gamma]$. Tensoring with the nonzero vector $[\gamma]$ is injective on $\mathcal W$, so its image has dimension $m$. In degree one, K\"unneth identifies the induced map with $[u]\mapsto f_{*,0}[u]\otimes[\gamma]$, whose kernel has dimension $\kappa$. The cone exact sequence therefore yields $kk_A-m+\kappa$ remaining logical dimensions.
\end{proof}

The selective cone is generally not the complete tensor product, so Corollary~\ref{cor:surgery-code-distance} does not apply to it automatically. Its connection weights must also be assessed using Eq.~\eqref{eq:selective-connection-weights}.

For example, take $\mathcal{W}=\operatorname{span}\{[\ell_{Z,\mu}+\ell_{Z,\nu}]\}$ and $\gamma=\ell^A_{Z,a}$. The selective lift imposes $\bar Z_{\mu,a}\bar Z_{\nu,a}$ without imposing the corresponding constraint at every other label. The complete lift imposes the constraint for every amplifier basis class.

To illustrate the different numbers of remaining logical qubits, take $\QQQ=\AAA=\code{4,2,2}$ with one $X$ check and one $Z$ check, and let $f$ add one nontrivial base $Z$ logical using $U_1=\mathbb F_2$ and $U_0=U_{-1}=0$. Then $m=1$ and $\kappa=0$. The amplified memory has $4\cdot4+1\cdot1+1\cdot1=18$ data qubits and four logical qubits. The base cone has four data qubits, one $X$ check, and two $Z$ checks, so its complete lift has $4\cdot4+2\cdot1+1\cdot1=19$ data qubits and two logical qubits. The selective cone adds one $Z$ check and no data qubits, leaving three logical qubits on the original 18 data qubits. This checks the distinction between the promoted subspaces; it does not establish fault tolerance of a direct measurement of the added generator.

Neither cone construction alone establishes a fault-tolerant measurement protocol. Ancilla preparation, syndrome extraction, merge and split boundaries, and the recorded logical outcome each require a fault analysis. For a merged code with an admissible extraction circuit and the required fault-response certificate, the memory theorems bound its SJ tensor extraction; an attained certificate gives the corresponding exact multiplication. These statements do not bound faults during preparation or deformation. Finally, measuring base operators independently on the central slices can reveal more information than measuring their product and can fail to preserve the amplified code. The inherited addresses and chain maps specify the intended logical constraints without justifying such independent measurements.

\section{Circuit-level simulations}
\label{app:bb18-simulations}

To assess the logical-error suppression provided by recursive distance amplification, we perform circuit-level Monte Carlo simulations of the $\code{18,4,4}$ bivariate bicycle code and its tensor constructions with the flagged $\code{4,2,2}$ amplifier. We consider four levels, $t=0,1,2,3$, with $(n,k)=(18,4),(90,8),(468,16),(2484,32)$, respectively. The seed uses the specified unflagged syndrome-extraction circuit, while each amplified circuit includes both newly introduced and inherited flags.

We use a circuit-level Pauli noise model with physical error probability $p$. Each CNOT is followed by two-qubit depolarizing noise with total probability $p$. Ancilla preparation in $|0\rangle$ or $|+\rangle$ is followed by an $X$ or $Z$ error, respectively, with probability $p$, and each syndrome or flag measurement outcome flips with probability $p$. Faults at distinct circuit locations are independent, and idling errors are neglected.

Each shot starts from an ideal encoded input and consists of $R=8$ noisy syndrome-extraction rounds followed by two noiseless rounds. Check detectors are formed from consecutive syndrome differences, with the first round compared against the ideal initial syndrome. Individual flag outcomes are included in the detector record without postselection. Both logical Pauli components are tracked in each shot to determine whether either CSS sector suffers a logical failure. We decode the two CSS sectors independently using belief propagation with localized statistics decoding (BP+LSD)~\cite{hillmann_localized_2025}. We use parallel min-sum BP with at most $60$ iterations and scaling factor $0.625$ for $t=0$, and at most $120, 240, 480$ iterations and scaling factor $0.75$ for $t=1,2,3$. The LSD settings are \texttt{LSD\_CS}, order $2$, and \texttt{bits\_per\_step}=8 throughout.

A shot is counted as a failure if either sector has a nonzero residual logical Pauli vector or returns a correction inconsistent with its input syndrome. Let $F$ denote the number of failed shots among $N$ samples. We report the effective logical error rate per logical qubit per noisy extraction round,
\begin{equation}
     \widehat q=\frac{F}{N},
     \qquad
     \widehat p_L=1-(1-\widehat q)^{1/(kR)}.
     \label{eq:bb18-joint-rate}
\end{equation}
Here $\widehat q$ estimates the joint block failure probability: a shot contributes once to $F$ even when multiple logical qubits or both CSS sectors fail. Error bars indicate $1\sigma$ statistical uncertainties in the estimated logical error rates.

The circuit resources at each amplification level are summarized in Table~\ref{tab:bb18-circuit-resources}. Simulations use Stim~\cite{gidney2021stim} for circuit sampling and \texttt{ldpc.BpLsdDecoder}~\cite{hillmann_localized_2025} for decoding. Sampling records are retained for subsequent statistical analysis.

\begin{table}[t]
\centering
\caption{Circuit resources for the $\code{18,4,4}$ seed and its recursive flagged-$\code{4,2,2}$ amplification. Qubit counts assume no optimization of ancilla reuse. CNOT counts and layer counts refer to one syndrome-extraction round, excluding preparation and measurement slots.}
\label{tab:bb18-circuit-resources}
\small
\setlength{\tabcolsep}{4pt}
\renewcommand{\arraystretch}{1.1}
\begin{tabular}{lrrrr}
\toprule
Resource & $t=0$ & $t=1$ & $t=2$ & $t=3$ \\
\midrule
Syndrome auxiliary systems      & 18  & 108 & 612  & 3384  \\
Flag auxiliary systems         & 0   & 36  & 324  & 2232  \\
Total physical qubits & 36  & 234 & 1404 & 8100  \\
CNOTs per round       & 108 & 828 & 5580 & 35316 \\
CNOT layers per round & 8   & 25  & 56   & 115   \\
\bottomrule
\end{tabular}
\end{table}

\section{Circuit distance versus fault-response certificates: an unflagged four-qubit example}
\label{app:unflagged422-certificate}

An unflagged $\code{4,2,2}$ circuit illustrates why circuit distance cannot generally replace the fault-response certificate in the amplification theorem. We use the fault model of Appendix~\ref{app:circuit_distance_amplification}, with ideal encoded input, noiseless idles, perfect terminal checks, and the complete noisy measurement record. Let $D_P^{\max}(\mathsf S)$ denote the largest certificate allowed by Definition~\ref{def:joint-response-certificate}. Circuit distance is defined by Eq.~\eqref{eq:iva-circuit-distance}; the number of noisy rounds is specified below.

\begin{figure*}[t]
 \centering
 \includegraphics[width=0.95\textwidth]{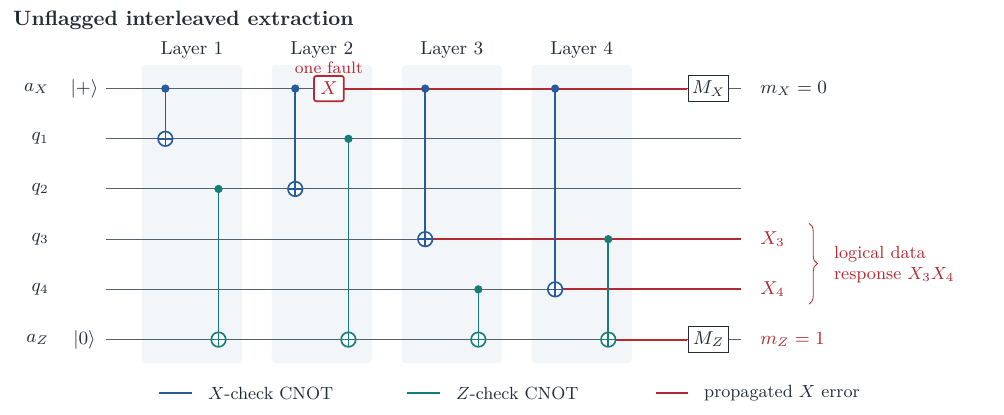}
 \caption{\textbf{A detected single fault with a logical data response.} An $X_{a_X}$ fault after the second $X$-check CNOT produces $X_3X_4$, noisy outcomes $(m_X,m_Z)=(0,1)$, and zero terminal syndrome. Gates within each layer act on disjoint qubits simultaneously.}
 \label{fig:uf422-circuit}
\end{figure*}

The circuit $\mathsf S_4$ measures
\begin{equation}
     S_X=X_1X_2X_3X_4,
     \qquad S_Z=Z_1Z_2Z_3Z_4.
     \label{eq:uf422-stabilizers}
\end{equation}
Prepare $a_X$ in $|+\rangle$ and $a_Z$ in $|0\rangle$. Apply $\operatorname{CX}(a_X,q_j)$ in the order $(q_1,q_2,q_3,q_4)$ and $\operatorname{CX}(q_j,a_Z)$ in the order $(q_2,q_1,q_4,q_3)$, pairing gates at the same position into four parallel layers (Fig.~\ref{fig:uf422-circuit}). Then measure $a_X$ in $X$ and $a_Z$ in $Z$. Reordering the checks into consecutive extraction blocks introduces two identical ancilla CNOTs that cancel by Eq.~\eqref{eq:common-cnot-exchange}. The ideal unitary $U$ therefore satisfies
\begin{equation}
     \begin{aligned}
     U^\dagger X_{a_X}U&=X_{a_X}S_X,\\
     U^\dagger Z_{a_Z}U&=Z_{a_Z}S_Z.
     \end{aligned}
     \label{eq:uf422-ideal-observables}
\end{equation}

An $X_{a_X}$ fault immediately after $\operatorname{CX}(a_X,q_2)$ propagates to
\begin{equation}
     E^{(4)}=X_{a_X}X_{a_Z}X_3X_4.
     \label{eq:uf422-hook-propagation}
\end{equation}
Its data response is a nontrivial logical, while $X_{a_Z}$ flips the noisy $Z$ outcome:
\begin{equation}
     \begin{aligned}
     (m_X,m_Z)&=(0,1),\\
     (s_{X,\mathrm{terminal}},s_{Z,\mathrm{terminal}})&=(0,0).
     \end{aligned}
     \label{eq:uf422-hook-record}
\end{equation}
This one-fault bundle is admissible in the $X$-check extraction gadget: there are no flags, and the certificate imposes no constraint on ordinary check outcomes. Thus $D_X^{\max}=1$. The dual fault $Z_{a_Z}$ after $\operatorname{CX}(q_1,a_Z)$ in layer 2 gives $Z_3Z_4$ with $(m_X,m_Z)=(1,0)$, establishing $D_Z^{\max}=1$.

To determine the circuit distance, a pure-$X$ fault $X_{a_X}^{\alpha}X_j^{\beta}$ after the $j$th $X$-check CNOT has data response
\begin{equation}
     e=\beta\mathbf e_j+\alpha\sum_{\ell=j+1}^{4}\mathbf e_\ell,
     \label{eq:uf422-suffix-response}
\end{equation}
where $\mathbf e_j$ is the unit vector at $q_j$ and $\alpha,\beta\in\{0,1\}$ are not both zero. This is a nontrivial logical only for $(j,\alpha,\beta)=(2,1,0)$ or $(3,1,1)$; both produce $X_3X_4$ and flip $m_Z$. Opposite-type CNOT faults affect at most one data qubit, and preparation or readout faults produce no nontrivial logical data response. The dual argument excludes single-fault zero-record histories in the $Z$ sector. This exclusion holds for any positive number of noisy rounds, since subsequent ideal rounds cannot erase an earlier recorded outcome or a nonzero data syndrome. Together with the upper bound $d_{\mathrm{circ},P}\le d_P=2$ from Appendix~\ref{app:circuit_distance_amplification}, this gives
\begin{equation}
     D_P^{\max}(\mathsf S_4)=1
     <d_{\mathrm{circ},P}(\mathsf S_4)=2,
     \qquad P\in\{X,Z\}.
     \label{eq:uf422-exact-circuit-distance}
\end{equation}

The same hook prevents $\mathsf S_4$ from directly replacing the flagged amplifier in our proof. Its response $v=(0,0,1,1)$ violates the local bound of Lemma~\ref{lem:iva-local-flag}:
\begin{equation}
     \min_{b\in\mathbb F_2}
     \operatorname{wt}\bigl(v+b(1,1,1,1)\bigr)=2>1.
     \label{eq:uf422-local-bound-failure}
\end{equation}
Ordinary check outcomes are unrestricted in this bound, and the flag condition is vacuous. Failure of this sufficient condition therefore does not preclude amplification for a particular unflagged circuit.

\begin{figure*}[t]
 \centering
 \includegraphics[width=0.7\textwidth]{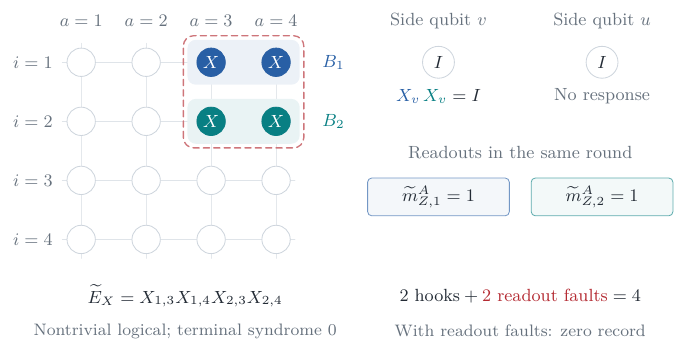}
 \caption{\textbf{Logical responses in an unflagged self-tensor circuit.} In one noisy round, an ancilla-$X$ fault after the second central $X$-check CNOT in each amplifier row $i=1,2$ produces $X_{i,3}X_{i,4}X_v$. The side responses cancel, leaving a logical rectangle and noisy outcomes $\widetilde m^A_{Z,1}=\widetilde m^A_{Z,2}=1$. Two additional readout faults hide these outcomes, giving a four-fault zero-record logical history. Colors distinguish the two hook responses.}
 \label{fig:uf422-tensor-mechanism}
\end{figure*}

For a concrete example, apply the SJ construction with $\mathsf S_4$ as both constituent circuits and no flags. The resulting $\mathsf S_{18}$ has central data $q_{i,a}$, $i,a\in\{1,2,3,4\}$, and side data qubits $v\in Q_2\otimes A_0$ and $u\in Q_0\otimes A_2$. The checks in Eqs.~\eqref{eq:tensor-HX} and~\eqref{eq:tensor-HZ} give eight retained checks of rank seven per sector, yielding $\code{18,4,4}$. The base and amplifier windows run $\mathsf S_4$ on columns and rows, with auxiliary systems $b_{P,a}$ and $a_{P,i}$, respectively. Side~1 applies $\operatorname{CX}(b_{X,a},u)$ and $\operatorname{CX}(v,b_{Z,a})$; side~2 applies $\operatorname{CX}(a_{X,i},v)$ and $\operatorname{CX}(u,a_{Z,i})$, visiting coordinates in increasing order within each family. Ancilla preparation and measurement bases are inherited from $\mathsf S_4$, with all readouts after side~2, giving 80 CNOTs per round.

The two hooks in Fig.~\ref{fig:uf422-tensor-mechanism} have combined central support $(\mathbf e_1+\mathbf e_2)\otimes(\mathbf e_3+\mathbf e_4)$ and identity on both side data qubits. They therefore give an admissible logical decomposition with assigned fault count two; two additional readout faults give a zero-record history with four faults. The dual construction gives the same upper bounds in the $Z$ sector.

For one, two, and eight noisy rounds, exhaustive enumeration of CNOT, preparation, and readout faults excludes zero-record logical histories with at most three faults in either sector. The enumeration retains the complete noisy record, terminal syndrome, and data response modulo stabilizers. A separate response enumeration, including singleton terms but without ordinary record constraints, excludes logical decompositions with assigned fault count one. Hence, for these round counts,
\begin{equation}
     \begin{aligned}
     D_X^{\max}(\mathsf S_{18})
     &=D_Z^{\max}(\mathsf S_{18})=2,\\
     d_{\mathrm{circ},X}(\mathsf S_{18})
     &=d_{\mathrm{circ},Z}(\mathsf S_{18})=4.
     \end{aligned}
     \label{eq:uf422-product-exact}
\end{equation}
This example establishes one-step circuit-distance doubling for the verified round counts, without implying recursive doubling.

\end{document}

%% file: newcommands.tex
\renewcommand{\ker}{\mathrm{ker}\text{ }}

\newcommand{\AAA}{\mathcal{A}}

\newcommand{\TTT}{\mathcal{T}}

\newcommand{\QQQ}{\mathcal{Q}}

\newcommand{\ox}{\otimes}
\newcommand{\x}{\times}

\newcommand{\tand}{\text{ and }}

\newcommand{\code}[1]{[\![#1]\!]}